\documentclass[11pt,a4paper]{article}
\usepackage{indentfirst,mathrsfs}
\usepackage{amsfonts,amsmath,amssymb,amsthm}
\usepackage{latexsym,amscd}
\usepackage{amsbsy}
\usepackage{color}
\usepackage[noadjust]{cite}
\usepackage{multirow}
\usepackage{makecell}
\usepackage{float}

\newtheorem{thm}{Theorem}
\newtheorem{lem}{Lemma}

\newtheorem{prop}{Proposition}

\newtheorem{example}{Example}
\newtheorem{remark}{Remark}

\begin{document}
\title{Several classes of bent functions over finite fields}
\author{Xi Xie, Nian Li, Xiangyong Zeng, Xiaohu Tang, Yao Yao
\thanks{The authors are with the Hubei Key Laboratory of Applied Mathematics, Faculty of Mathematics and Statistics, Hubei
University, Wuhan, 430062, China. Xiaohu Tang is also with the Information Security and National Computing Grid Laboratory, Southwest Jiaotong University, Chengdu, 610031, China. Email: xi.xie@aliyun.com, nian.li@hubu.edu.cn, xzeng@hubu.edu.cn, xhutang@swjtu.edu.cn, yyao1966@163.com}
}
\date{}
\maketitle
\begin{quote}
{{\bf Abstract:} Let $\mathbb{F}_{p^{n}}$ be the finite field with $p^n$ elements and $\operatorname{Tr}(\cdot)$ be the trace function from $\mathbb{F}_{p^{n}}$ to $\mathbb{F}_{p}$, where $p$ is a prime and $n$ is an integer. Inspired by the works of Mesnager (IEEE Trans. Inf. Theory 60(7): 4397-4407, 2014) and Tang et al. (IEEE Trans. Inf. Theory 63(10): 6149-6157, 2017), we study a class of bent functions of the form $f(x)=g(x)+F(\operatorname{Tr}(u_1x),\operatorname{Tr}(u_2x),\cdots,\operatorname{Tr}(u_{\tau}x))$, where $g(x)$ is a function from $\mathbb{F}_{p^{n}}$ to $\mathbb{F}_{p}$,  $\tau\geq2$ is an integer, $F(x_1,\cdots,x_n)$ is a reduced polynomial in $\mathbb{F}_{p}[x_1,\cdots,x_n]$ and $u_i\in \mathbb{F}^{*}_{p^n}$ for $1\leq i \leq \tau$. As a consequence, we obtain a generic result on the Walsh transform of $f(x)$ and characterize the bentness of $f(x)$ when $g(x)$ is bent for $p=2$ and $p>2$ respectively. Our results generalize some earlier works. In addition, we study the construction of bent functions $f(x)$ when $g(x)$ is not bent for the first time and present a class of bent functions from non-bent Gold functions.
}

\end{quote}

\section{Introduction}
Boolean bent functions were first introduced by Rothaus in 1976 \cite{R} as an interesting combinatorial object with maximum Hamming distance to the set of all affine functions. Over the last four decades, bent functions have attracted a lot of research interest due to their important applications in cryptography \cite{C2010}, sequences \cite{OSW} and coding theory \cite{CHLL,DFZ}. Kumar, Scholtz and Welch in \cite{KSW} generalized the notion of Boolean bent functions to the case of functions over an arbitrary finite field.

Given a function $f(x)$ mapping from $\mathbb{F}_{p^{n}}$ to $\mathbb{F}_{p}$, the Walsh transform of $f(x)$ is defined by
$$\widehat{f}(b)=\sum\nolimits_{x\in\mathbb{F}_{p^n}}{\omega}^{f(x)-\operatorname{Tr}(bx)}, \, b\in\mathbb{F}_{p^n},$$
where $\omega=e^{\frac{2\pi\sqrt{-1}}{p}}$ is a complex primitive $p$-th root of unity.
According to \cite{KSW}, $f(x)$ is called a $p$-ary bent function if all its Walsh coefficients satisfy $\big|\widehat{f}(b) |=p^{n/2}$. A $p$-ary bent function $f(x)$ is called regular if $\widehat{f}(b)=p^{n/2}\omega^{\widetilde{f}(b)}$ holds for some function $\widetilde{f}(x)$ mapping $\mathbb{F}_{p^n}$ to $\mathbb{F}_p$, and it is called weakly regular if there exists a complex $\mu$ having unit magnitude such that $\widehat{f}(b)=\mu^{-1}p^{n/2}\omega^{\widetilde{f}(b)}$ for all $b\in\mathbb{F}_{p^n}$. The function $\widetilde{f}(x)$ is called the dual of $f(x)$ and it is also bent.

An interesting class of bent functions over finite fields with the form
\begin{equation}\label{f=g+F}
f(x)=g(x)+F(\operatorname{Tr}(u_1x),\operatorname{Tr}(u_2x),\cdots,\operatorname{Tr}(u_{\tau}x))
\end{equation}
was studied in the past years, where $g(x)$ is a function from $\mathbb{F}_{p^{n}}$ to $\mathbb{F}_{p}$, $F(x_1,\cdots,x_{\tau})$ is an arbitrary reduced polynomial in $\mathbb{F}_p[x_1,\cdots,x_{\tau}]$, $\tau\geq2$ is an integer and $u_i\in \mathbb{F}^{*}_{p^n}$ for all $1\leq i \leq \tau$. The initial work on this issue is due to Mesnager \cite{M} who studied the case $p=2$, $F(x_1,x_2)=x_1x_2$ and $g(x)$ is a bent function whose dual function has a null second order derivative. This motivated Xu et al. to construct bent functions by using some known bent functions $g(x)$ via the cases $F(x_1,x_2,x_3)=x_1x_2x_3$ for $p=2$ \cite{XCX} and $F(x_1,x_2)=x_1x_2$ for $p=3$ \cite{XCX1,XCX2}, respectively. Later, Wang et al. characterized the bentness of $f(x)$ when $g(x)$ is bent and $F(x_1,\cdots,x_{\tau})=x_1\cdots x_{\tau}$ for $p=2$ and consequently constructed bent functions of the form \eqref{f=g+F} from some bent functions $g(x)$ whose dual functions are already known \cite{WWL}. Meanwhile, also for $p=2$, Tang et al. \cite{T} investigated the bentness of $f(x)$ with the form \eqref{f=g+F} for an arbitrary reduced polynomial $F(x_1,\cdots,x_{\tau})$ and a bent function $g(x)$ whose dual satisfies
\begin{equation}\label{g}
\widetilde{g}(x-\sum\nolimits_{i=1}^{\tau}u_it_i)=\widetilde{g}(x)+\sum\nolimits_{i=1}^{\tau}g_i(x)t_i,
\end{equation}
where $t_i\in\mathbb{F}_{p}$ and $g_{i}(x)$ is a function from $\mathbb{F}_{p^{n}}$ to $\mathbb{F}_{p}$ for each $1\leq i \leq \tau$. The analogues of the results in \cite{T} for an odd prime $p$ were obtained in \cite{QTZF} where $g(x)$ was required to be a homogeneous quadratic bent function and its dual also satisfies \eqref{g}. In 2019, Zheng et al. \cite{ZPKL} showed that for $p=2$, $f(x)$ of the form \eqref{f=g+F} is bent for any reduced polynomial $F(x_1,\cdots,x_{\tau})$ if and only if $g(x)$ is bent whose dual satisfies \eqref{g}.

Inspired by the above works, in this paper, we further study the construction of bent functions with the form \eqref{f=g+F}. We first derive a generic result on Walsh transform of the function $f(x)$ of the form \eqref{f=g+F} in which $g(x)$ is not necessarily bent. Then we characterize the bentness of $f(x)$ in \eqref{f=g+F} when $g(x)$ is bent whose dual satisfies
\begin{equation}\label{g-Aij}
\widetilde{g}(x-\sum\nolimits_{i=1}^{\tau}u_it_i)=\widetilde{g}(x)+\sum\nolimits_{1\leq i\leq j\leq \tau}A_{ij}t_it_j+\sum\nolimits_{i=1}^{\tau}g_i(x)t_i
\end{equation}
for some $A_{ij}\ne 0$. Our results generalize some earlier works in this direction. In addition, we attempt to construct bent functions $f(x)$ having the form \eqref{f=g+F} from non-bent functions $g(x)$ and consequently obtain a class of such bent functions by using non-bent Gold functions and $F(x_1,x_2)=x_1x_2$. To the best of our knowledge, the construction of bent functions with the form \eqref{f=g+F} from either bent functions satisfying \eqref{g-Aij} for some $A_{ij}\ne 0$ or non-bent functions is studied in this paper for the first time in the literature.

The rest of this paper is organized as follows. Section \ref{prel} gives some preliminaries. Section \ref{Walsh} provides some results on the Walsh transform of functions with the form \eqref{f=g+F} and characterizes the bentness of such functions when $g(x)$ satisfies \eqref{g-Aij} for some nonzero $A_{ij}$. Section \ref{cons2} proposes a class of bent functions of the form \eqref{f=g+F} in which $g(x)$ is a non-bent Gold function and $F(x_1,x_2)=x_1x_2$. Section \ref{conc} concludes this paper.

\section{Preliminaries}\label{prel}
Throughout this paper, let $\mathbb{F}_{p^n}$ denote the finite field with $p^n$ elements, where $n$ is a positive integer and $p$ is a prime. The trace function from
$\mathbb{F}_{p^n}$ to its subfield $\mathbb{F}_{p^k}$ is defined by $\operatorname{Tr}_{k}^{n}(x)=\sum_{i=0}^{n/k-1}x^{p^{ik}}$. In particular, when $k=1$, we use the notation $\operatorname{Tr}(x)$ instead of $\operatorname{Tr}_{1}^{n}(x)$.

\subsection{Algebraic degree}
A function $F(x_1,\cdots,x_n):\mathbb{F}_{p}^{n}\mapsto\mathbb{F}_p$ is often represented by its algebraic normal form:
\begin{equation}\label{F(x_1,...,x_n)}
F(x_1,\cdots,x_n)=\sum_{e=(e_1,\cdots,e_n)\in\mathbb{F}_{p}^n}a(e) (\prod_{i=1}^nx_i^{e_i} ),\,\,\,a(e)\in \mathbb{F}_{p}.
\end{equation}
A polynomial in $\mathbb{F}_p[x_1,\cdots,x_n]$ with the form \eqref{F(x_1,...,x_n)} is called a reduced polynomial.
The algebraic degree of $F(x_1,\cdots,x_{n})$, denoted by ${\rm deg}(F)$, is defined as
${\rm deg}(F)={\rm max}_{e\in\mathbb{F}_{p}^n}\{\sum_{i=1}^n e_i: a(e)\neq 0\}$, where $e=(e_1,\cdots,e_n)$. The following lemma will be used to determine the algebraic degree of some reduced functions, which is a direct generalization of the result proposed in \cite[Lemma 2.1]{T}.

\begin{lem}\label{algebraic degree}
Let $u_{1}, u_{2}, \cdots, u_{\tau}\in\mathbb{F}_{p^{n}}^{*}$ be linearly independent over $\mathbb{F}_{p},$ where $\tau$ is an integer with $1\leq\tau\leq n.$ Let $F(x_1,\cdots,x_{\tau})$ be a reduced polynomial in $\mathbb{F}_p[x_1,\cdots,x_{\tau}]$ of algebraic degree $d$. Then the following univariate function
$$F (\operatorname{Tr}(u_1x),\operatorname{Tr}(u_2x),\cdots,\operatorname{Tr}(u_{\tau}x) )$$
has algebraic degree d.
\end{lem}

The algebraic degree of a $p$-ary bent function has been characterized as follows.

\begin{lem}{\rm(\cite[Propositions 4.4 and 4.5]{H})}\label{algebraic degree. limit}
Let $f(x)$ be a bent function from $\mathbb{F}_{p^n}$ to $\mathbb{F}_{p}$, then the algebraic degree
$\deg(f)$ of $f(x)$ satisfies $\deg(f)\leq \frac{(p-1)n}{2}+1$, and if $f(x)$ is weakly regular bent, then $\deg(f)\leq \frac{(p-1)n}{2}$.
\end{lem}

\subsection{ Certain exponential sums}
For each $b\in\mathbb{F}_{p^n}$, the function $\chi_b(x)=\omega^{{\rm Tr}(bx)}$ defines an additive character for $x\in\mathbb{F}_{p^n}$. The character $\chi:=\chi_1$ is called the canonical additive character of $\mathbb{F}_{p^n}$.

\begin{lem}{\rm(\cite[Theorems 5.15 and 5.33]{LN})} \label{char sum(ax^2+bx+c)}
Let $p$ be an odd prime, $\eta$ be the quadratic multiplicative character of $\mathbb{F}_{p}$ and $p^{*}=(-1)^{\frac{p-1}{2}}p=\eta(-1)p$. Then
$$\sum\nolimits_{x\in\mathbb{F}_{p}}\chi (ax^2+bx )=\eta(a)\sqrt{p^{*}}\omega^{-\frac{b^2}{4a}},\, a\in\mathbb{F}_{p}^{*},\,b\in\mathbb{F}_{p}.$$
\end{lem}

\begin{lem}\label{expH.k=2}
Let $p$ be an odd prime and $a_i\in\mathbb{F}_{p}$ for $i=1,2,3,4,5$. Denote
$$H=\sum\nolimits_{(x, y)\in\mathbb{F}_{p}^2}\omega^{a_1 x^{2}+a_2 y^{2}+a_3 xy+a_4 x+a_5 y}.$$
If $a_3^2-4a_1a_2=0$, then
$$H=\left \{\begin{array}{lll}
p^2, & {\rm if} \,\, a_1=0,\,a_2=a_4=a_5=0,\\[0.05in]
\eta(a_2)p\sqrt{p^*}\omega^{-\frac{a_5^2}{4a_2}}, &{\rm if} \,\, a_1=0,\, a_2\neq0,\,a_4=0,\\[0.05in]
\eta(a_1)p\sqrt{p^*}\omega^{-\frac{a_4^2}{4a_1}}, &{\rm if} \,\, a_1\neq0,\, a_5=\frac{a_3a_4}{2a_1},\\[0.05in]
0, &{\rm otherwise},
\end{array}\right .$$
and if $a_3^2-4a_1a_2\neq0$, then
$$H=\eta (a_3^2-4a_1a_2 )p\omega^{\frac{a_2a_4^2+a_1a_5^2-a_3a_4a_5}{a_3^2-4a_1a_2}}.$$
\end{lem}

\begin{proof}
We first consider the case $a_3^2-4a_1a_2=0$. If $a_1=0$, then $a_3=0$ and Lemma \ref{char sum(ax^2+bx+c)} yields
$$H=\sum\limits_{(x, y)\in\mathbb{F}_{p}^2}\omega^{a_4 x+a_2 y^{2}+a_5 y}
= \left\{\begin{array}{lll}
p^2, & {\rm if} \,\, a_2=a_4=a_5=0,\\[0.05in]
p\sqrt{p^*}\eta(a_2)\omega^{-\frac{a_5^2}{4a_2}}, &{\rm if} \,\, a_2\neq0,\,a_4=0,\\[0.05in]
0, &{\rm otherwise}.\end{array}\right .$$
If $a_1\neq 0$, again by Lemma \ref{char sum(ax^2+bx+c)} one obtains
\begin{eqnarray*}
  H = \sqrt{p^*}\eta(a_1)\sum\nolimits_{y\in\mathbb{F}_{p}} \omega^{a_2 y^{2}+a_5 y-\frac{(a_3y+a_4)^2}{4a_1}}
    = \sqrt{p^*}\eta(a_1)\sum\nolimits_{y\in\mathbb{F}_{p}} \omega^{(a_5-\frac{a_3a_4}{2a_1}) y-\frac{a_4^2}{4a_1}}
\end{eqnarray*}
due to $a_3^2-4a_1a_2=0$. This implies that $$H=p\sqrt{p^*}\eta(a_1)\omega^{-\frac{a_4^2}{4a_1}}$$
if $a_5=a_3a_4/2a_1$ and $H=0$ otherwise.

Next, we calculate $H$ for the case $a_3^2-4a_1a_2\neq0$. If $a_1=0$, then $a_3\ne 0$ and it can be readily verified that
$$H=\sum\nolimits_{(x, y)\in\mathbb{F}_{p}^2}\omega^{(a_3y+a_4) x+a_2 y^{2}+a_5 y}=p\omega^{a_2(-\frac{a_4}{a_3})^2+a_5(-\frac{a_4}{a_3}) }.$$
If $a_1\ne0$, then by Lemma \ref{char sum(ax^2+bx+c)}, one gets
\begin{eqnarray*}
  H = \sqrt{p^*}\eta(a_1)\sum\nolimits_{y\in\mathbb{F}_{p}} \omega^{a_2 y^{2}+a_5 y-\frac{(a_3y+a_4)^2}{4a_1}}
    = p^*\eta (4a_1a_2-a_3^2 )\omega^{\frac{a_1}{a_3^2-4a_1a_2} (a_5-\frac{a_3a_4}{2a_1} )^2
    -\frac{a_4^2}{4a_1}}.
\end{eqnarray*}
Then the result follows from $p^{*}=\eta(-1)p$ and $p^*\eta (4a_1a_2-a_3^2 )=p\eta (a_3^2-4a_1a_2 )$. This completes the proof.
\end{proof}

\subsection{Walsh transform of Gold functions}\label{gold}

In this subsection, we present some known results on the Walsh transform of Gold functions which will be used to construct bent functions $f(x)$ of the form \eqref{f=g+F} from non-bent functions $g(x)$ in Section \ref{cons2}.

Coulter studied the Walsh transform of Gold functions for odd $p$ \cite{CRS.1998} and $p=2$ \cite{CRS.1999} respectively. Let $n, k$ be positive integers with $d=\gcd(k,n)$ and $g(x)=\operatorname{Tr}(ax^{p^k+1})$, where $a\in\mathbb{F}_{p^n}^{*}$, then the Walsh transform of $g(x)$ has been determined as below.

\begin{lem}{\rm(\cite[Theorem 4.2]{CRS.1999})}\label{W(g).odd}
Suppose $n/d$ is odd. Then for $b\in\mathbb{F}_{2^n}$,
$$\widehat{g}(b)=\left \{\begin{array}{ll}
0, & {\rm if} \,\, \operatorname{Tr}^{n}_{d}(bc^{-1})\neq1,\\
\pm2^{\frac{n+d}{2}}, &{\rm if} \,\, \operatorname{Tr}^{n}_{d}(bc^{-1})=1,
\end{array} \right.$$
 where $c\in\mathbb{F}_{2^n}^{*}$ is the unique element satisfying $c^{2^k+1}=a$.
\end{lem}

\begin{lem}{\rm(\cite[Theorems 1 and 2]{CRS.1998.1,CRS.1998} and \cite[Theorem 5.3]{CRS.1999})}\label{W(g).even}
Suppose $n/d$ is even, $n=2m$  and $p$ is a prime. Then for $b\in\mathbb{F}_{p^n}$,
$$\widehat{g}(b)= \left\{\begin{array}{lll}
(-1)^{m/d}p^{m}\overline{\chi}(ax_0^{p^k+1}), & {\rm if}\,\,a^{\frac{p^n-1}{p^d+1}}\not=(-1)^{\frac{m}{d}},\\
(-1)^{m/d+1}p^{m+d}\overline{\chi}(ax_0^{p^k+1}), & {\rm if}\,\,a^{\frac{p^n-1}{p^d+1}}=(-1)^{\frac{m}{d}},\, \operatorname{Tr}^n_{2d}(b/ac^{p^k})=0,\\
0, & {\rm otherwise},
\end{array} \right.$$
where $c\in\mathbb{F}_{p^n}^{*}$ satisfies $a^{p^k}c^{p^{2k}}+ac=0$ and $x_0$ is the solution of $a^{p^k}x^{p^{2k}}+ax=-b^{p^k}$.
\end{lem}


\section{Constructions of bent functions of the form \eqref{f=g+F}}\label{Walsh}

In this section, we first derive a generic result on the Walsh transform of $f(x)$ with the form \eqref{f=g+F} in which $g(x)$ is not necessarily bent. Then, we characterize the bentness of $f(x)$ in \eqref{f=g+F} for a bent function  $g(x)$ whose dual satisfies \eqref{g-Aij} with some $A_{ij}\ne0$ for $p=2$ and $p>2$ respectively.


 The Walsh transform of a multivariate function $F(x_1,\cdots,x_{n})$ over $\mathbb{F}_{p}^n$ is
\begin{equation}\label{F1}
\widehat{F}(b_1,\cdots,b_{n})=\sum_{(x_1,\cdots,x_{n})\in\mathbb{F}_{p}^n}
\omega^{F(x_1,\cdots,x_{n})-\sum_{i=1}^nb_ix_i},
\end{equation}
where $(b_1,\cdots,b_{n})\in\mathbb{F}_{p}^n$. Then the inverse Walsh transform of $F(x_1,\cdots,x_{n})$ is given by
\begin{equation}\label{F2}
\omega^{F(x_1,\cdots,x_{n})}=\frac{1}{p^n}\sum_{(b_1,\cdots,b_{n})\in\mathbb{F}_{p}^n}
\omega^{\sum_{i=1}^nb_ix_i}\widehat{F}(b_1,\cdots,b_{n}).
\end{equation}

Using \eqref{F1} and \eqref{F2}, the Walsh transform of $f(x)$ in \eqref{f=g+F} can be expressed as below.

\begin{thm}\label{thm.wf.f=g+F}
Let $f(x)$ be defined as \eqref{f=g+F}, then for any $b\in \mathbb{F}_{p^n}$,
\begin{eqnarray*}
\widehat{f}(b)=\frac{1}{p^{\tau}}\sum_{(t_1,\cdots,t_{\tau})\in\mathbb{F}_{p}^{\tau}}\widehat{F}(t_1,\cdots,t_{\tau})
\widehat{g}(b-\sum\nolimits_{i=1}^{\tau}t_iu_i).
\end{eqnarray*}
In particular, if $F(x_1,\cdots,x_{\tau})=x_1\cdots x_{\tau}$, then
\begin{eqnarray*}
\widehat{f}(b)=\frac{1}{p^{\tau-1}}\sum^{\tau-1}_{i,j=1} \sum^{p-1}_{t_{i}=0} \sum^{p-1}_{x_j=0}
\omega^{-\sum^{\tau-1}_{i=1}x_it_{i}}\widehat{g}(b-\sum\nolimits^{\tau-1}_{i=1}t_{i}u_i
-\prod\nolimits^{\tau-1}_{i=1}x_iu_\tau).
\end{eqnarray*}
\end{thm}
\begin{proof}
According to \eqref{F2}, for any $b\in \mathbb{F}_{p^n}$, one obtains
$$\begin{aligned} \widehat{f}(b) &=\sum_{x \in \mathbb{F}_{p^n}}\omega^{g(x)+F (\operatorname{Tr}(u_1x),\operatorname{Tr}(u_2x),
\cdots,\operatorname{Tr}(u_{\tau}x) )-\operatorname{Tr}(bx)}
\\ &=
\frac{1}{p^{\tau}}\sum_{(t_1,\cdots,t_{\tau})\in\mathbb{F}_{p}^{\tau}}\sum_{x \in \mathbb{F}_{p^n}}\omega^{g(x)-\operatorname{Tr}(bx)+
\sum_{i=1}^{\tau}\operatorname{Tr}(u_{i}x)t_i}\widehat{F}(t_1,\cdots,t_{\tau})
\\ &=
\frac{1}{p^{\tau}}\sum_{(t_1,\cdots,t_{\tau})\in\mathbb{F}_{p}^{\tau}}\widehat{F}(t_1,\cdots,t_{\tau})
\widehat{g}(b-\sum\nolimits_{i=1}^{\tau}t_iu_i).
\end{aligned}$$

If $F(x_1,\cdots,x_{\tau})=x_1\cdots x_{\tau}$, then by \eqref{F1}, one gets
$$
\begin{aligned} \widehat{f}(b) &=\frac{1}{p^{\tau}}\sum_{(t_1,\cdots,t_{\tau})\in\mathbb{F}_{p}^{\tau}}
\sum_{(x_1,\cdots,x_{\tau})\in\mathbb{F}_{p}^{\tau}}\omega^{x_1\cdots x_{\tau}-\sum_{i=1}^{\tau}t_ix_i}
\widehat{g}(b-\sum\nolimits_{i=1}^{\tau}t_iu_i)
\\ &= \frac{1}{p^{\tau}} \sum_{i,j=1}^{\tau-1} \sum_{t_{i}=0}^{p-1} \sum_{x_{j}=0}^{p-1} \sum_{t_{\tau}=0}^{p-1} \sum_{x_{\tau}=0}^{p-1}
\omega^{(\prod\nolimits_{i=1}^{\tau-1}x_{i} -t_{\tau})x_{\tau}-\sum_{i=1}^{\tau-1}t_ix_i} \widehat{g}(b-\sum\nolimits_{i=1}^{\tau}t_iu_i)
\\ &=
\frac{1}{p^{\tau-1}} \sum_{i,j=1}^{\tau-1} \sum_{t_{i}=0}^{p-1}\sum_{x_{j}=0}^{p-1}  \omega^{-\sum_{i=1}^{\tau-1} x_{i} t_{i}} \widehat{g}(b-\sum\nolimits_{i=1}^{\tau-1} t_{i} u_{i}-\prod\nolimits_{i=1}^{\tau-1} x_{i} u_{\tau})
\end{aligned}
$$
due to the fact that $\sum_{x_{\tau}=0}^{p-1} \omega^{(\prod\nolimits_{i=1}^{\tau-1}x_{i} -t_{\tau})x_{\tau}}=0$ if $t_{\tau}\neq \prod\nolimits_{i=1}^{\tau-1}x_{i}$. This completes the proof.
\end{proof}

\begin{remark}
Note that Theorem \ref{thm.wf.f=g+F} holds for an arbitrary prime $p$ and $g(x)$ is a function from $\mathbb{F}_{p^{n}}$ to $\mathbb{F}_{p}$ which is not necessary to be bent. It generalizes some earlier works:
\begin{enumerate}
 \vspace{-2mm} \item [1).] Theorem \ref{thm.wf.f=g+F} is a generalization of Lemma 1 in \cite{WWL}.
\vspace{-2mm}  \item [2).] If one takes $p=2$, $\tau=2$ and $f(x)=g(x)+\operatorname{Tr}(ux)\operatorname{Tr}(vx)$, where $u,v\in\mathbb{F}_{2^n}^{*}$, then our result gives $\widehat{f}(b)=\frac{1}{2}(\widehat{g}(b)+\widehat{g}(b+u)+\widehat{g}(b+v)-\widehat{g}(b+u+v))$ for any $b\in\mathbb{F}_{2^n}$. When $g(x)$ is a Boolean bent function with the dual $\widetilde{g}(x)$, it can be verified that $f(x)$ is bent if and only if $\widetilde{g}(x)+\widetilde{g}(x+u)+\widetilde{g}(x+v)+\widetilde{g}(x+u+v)=0$.
      This is Corollary 5 in \cite{M}.
 \vspace{-2mm} \item [3).] If one takes $p=2$, $\tau=3$ and $f(x)=g(x)+\operatorname{Tr}(ux)\operatorname{Tr}(vx)\operatorname{Tr}(rx)$, where $u,v,r\in\mathbb{F}_{2^n}^{*}$, then our result gives $\widehat{f}(b)=\frac{1}{4}(3\widehat{g}(b)+\widehat{g}(b+u)+\widehat{g}(b+v)
+\widehat{g}(b+r)+\widehat{g}(b+u+v+r)-\widehat{g}(b+u+v)-\widehat{g}(b+u+r)-\widehat{g}(b+v+r))$ for any $b\in\mathbb{F}_{2^n}$. This is Lemma 1 in \cite{XCX}.
\vspace{-2mm}  \item [4).] If one takes $p=3$, $\tau=2$ and $f(x)=g(x)+\operatorname{Tr}(ux)\operatorname{Tr}(vx)$ for $u,v\in\mathbb{F}_{3^n}^{*}$, then by Theorem \ref{thm.wf.f=g+F} one obtains $\widehat{f}(b)=\frac{1}{3}(\widehat{g}(b)+\widehat{g}(b+u)+\widehat{g}(b-u)
+\widehat{g}(b-v)+\widehat{g}(b+v)+\omega\widehat{g}(b-v+u)+\omega^2\widehat{g}(b-v-u)+\omega^2\widehat{g}(b+v+u)+\omega\widehat{g}(b+v-u)$
for any $b\in\mathbb{F}_{3^n}$, where $\omega$ is a primitive $3$-rd root of unity.
This is exactly Lemma 4 in \cite{XCX1}.
\end{enumerate}
 \end{remark}

\subsection{Bent functions of the form \eqref{f=g+F} for $p=2$}\label{Walsh.g bent.p=2}

Let $p=2$ and $f(x)$ be defined as \eqref{f=g+F}. Tang et al. \cite{T}  proved that $f(x)$ is bent for any reduced polynomial $F(x_1,\cdots,x_{\tau})$ in $\mathbb{F}_p[x_1,\cdots,x_n]$ if $g(x)$ is a bent function whose dual satisfies \eqref{g}. Later, Zheng et al. \cite{ZPKL}  showed that this condition is also necessary. In this section, we consider the bentness of $f(x)$ for a more general bent function $g(x)$, i.e., $g(x)$ satisfies \eqref{g-Aij} for some $A_{ij}\neq0$, where $1\leq i< j\leq \tau$ and $\tau$ is a positive integer.

Suppose that $g(x)$ is a bent function over $\mathbb{F}_{2^{n}}$ and its dual satisfies \eqref{g-Aij}, then for $b\in\mathbb{F}_{2^{n}}$, $t_i\in\mathbb{F}_{2}$ and $u_i\in\mathbb{F}_{2^{n}}^*$, where $i=1,2,\cdots,\tau$,  one gets
 $$\widehat{g}(b-\sum\nolimits_{i=1}^{\tau}t_iu_i)
=2^{\frac{n}{2}}(-1)^{\widetilde{g}(b-\sum\nolimits_{i=1}^{\tau}t_iu_i)},$$
where $\widetilde{g}(x)$ is the dual of $g(x)$.
This together with \eqref{F1} and Theorem \ref{thm.wf.f=g+F} gives
\begin{eqnarray}\label{f^b.p=2}
\widehat{f}(b)&=&\frac{1}{2^{\tau}}\sum_{(t_1,\cdots,t_{\tau})\in\mathbb{F}_{2}^{\tau}}\widehat{F}(t_1,\cdots,t_{\tau})
\widehat{g}(b-\sum\nolimits_{i=1}^{\tau}t_iu_i)\nonumber\\
&=&2^{n/2-\tau}\sum_{i,j=1}^{\tau} \sum_{x_{i}=0}^{1}\sum_{t_{j}=0}^{1}
(-1)^{F(x_1,\cdots,x_{\tau})+\widetilde{g}(b)+\sum\limits_{1\leq i< j\leq \tau}A_{ij}t_it_j+\sum\limits_{i=1}^{\tau}  (g_i(b)+x_i+A_{ii} )t_i},
\end{eqnarray}
where $g_{i}(x)$ is a function from $\mathbb{F}_{2^{n}}$ to $\mathbb{F}_{2}$ for each $1\leq i \leq \tau$.

Now assume that the nonzero elements in $\{A_{ij}:1\leq i<j\leq \tau\}$ are $A_{i_1j_1}, \cdots, A_{i_{\ell}j_{\ell}}$, where $\ell$ is a positive integer, $i_s<j_{s}$ for $1\leq s\leq\ell$ and $1\leq i_1\leq \cdots\leq i_{\ell}< \tau$. For simplicity, denote $h_i=g_i(b)+A_{ii}$ and define
\begin{equation}\label{T}
\Gamma=\{i_1,j_1,\cdots,i_{\ell},j_{\ell}\}.
\end{equation}

Then, we discuss \eqref{f^b.p=2} as the following two cases:

Case I: $\#\Gamma=2\ell$, where $2\leq2\ell\leq \tau$.

For this case,  we have $A_{i_sj_s}=1$ for $1\leq s\leq\ell$ and $A_{ij}=0$ otherwise. Then \eqref{f^b.p=2} becomes
\begin{eqnarray*}
  \widehat{f}(b) &=&2^{n/2-\tau}(-1)^{\widetilde{g}(b)}\sum_{i,j=1}^{\tau} \sum_{x_{i}=0}^{1}\sum_{t_{j}=0}^{1}
(-1)^{F(x_1,\cdots,x_{\tau})+\sum_{s=1}^{\ell}t_{i_s}t_{j_s}+\sum_{i=1}^{\tau}  (h_i+x_i )t_i}  \\
    &=& 2^{n/2-2\ell}(-1)^{\widetilde{g}(b)}\sum_{i,j\in\Gamma} \sum_{x_{i}=0}^{1}\sum_{t_{j}=0}^{1}
(-1)^{F(x_1,\cdots,x_{\tau})|_{x_i=h_i, i\not\in \Gamma}
+\sum\limits_{s=1}^{\ell}(t_{j_s}+ (h_{i_s}+x_{i_s} ))t_{i_s}
+\sum\limits_{s=1}^{\ell} (h_{j_s}+x_{j_s} )t_{j_s}}\\
&=&2^{n/2-\ell}(-1)^{\widetilde{g}(b)}\sum_{i\in\Gamma} \sum_{x_{i}=0}^{1}
(-1)^{F(x_1,\cdots,x_{\tau})|_{x_i=h_i, i\not\in \Gamma}
+\sum\limits_{s=1}^{\ell}  (h_{i_s}+x_{i_s} ) (h_{j_s}+x_{j_s} )},
\end{eqnarray*}
where the second identity holds due to $\sum_{t_{i}=0}^{1} (-1)^{(h_i+x_i) t_{i}}=2$ if $x_i=h_i$ and $0$ otherwise for any $i\not\in \Gamma$.

Then, in this case, $f(x)$ defined by \eqref{f=g+F} can be bent for certain special $F(x_1,\cdots,x_{\tau})$.

\begin{thm}\label{p=2.F1}
Let $n=2m$ and $u_{1}, u_{2}, \cdots, u_{\tau}$ be pairwise distinct elements in $\mathbb{F}_{2^{n}}^{*}$, where $2\leq\tau\leq m$. Let $g(x)$ be a bent function over $\mathbb{F}_{2^n}$ whose dual satisfies \eqref{g-Aij} with $A_{ij}\neq 0$ for $i,j\in\Gamma$ and $A_{ij}=0$ otherwise, where $\Gamma$ is defined by \eqref{T}. If $\#\Gamma=2\ell$ and $F(x_1,\cdots,x_{\tau})$ satisfies
$$F(x_1,\cdots,x_{\tau})|_{x_i=h_i, i\not\in \Gamma}=
\sum\nolimits_{s=1}^{\ell}(F_{i_s}x_{i_s}+F_{j_s}x_{j_s})+F_0$$
for some $F_{i_s}, F_{j_s}, F_0\in \mathbb{F}_2$, where $h_i=g_i(b)+A_{ii}$, then $f(x)$ defined by \eqref{f=g+F} is bent and its Walsh transform at point $b\in\mathbb{F}_{2^n}$ is
$$\widehat{f}(b)=2^{n/2}(-1)^{\widetilde{g}(b)+\sum_{s=1}^{\ell}[ (h_{i_s}+F_{j_s} ) (h_{j_s}
+F_{i_s} )+h_{i_s}h_{j_s}]+F_0}.$$
\end{thm}

\begin{example}
Let $n=2m=6$, $\tau=4$, $\xi$ be a primitive element of $\in \mathbb{F}_{2^6}^{*}$ and $g(x)=\operatorname{Tr}_1^3(x^{9})$. From \cite{M} we know that the dual of $g(x)$ satisfies \eqref{g-Aij} with $A_{ii}=\operatorname{Tr}_1^3 (u_i^{9} )$, $A_{ij}=\operatorname{Tr} (u_iu_j^{8} )$ for $i<j$ and $g_i(x)=\operatorname{Tr} (u_i^{8}x )$. Taking $u_1=1$, $u_2=\xi$, $u_3=\xi^4$, $u_4=\xi^2$, and $F(x_1,x_2,x_3,x_4)=x_1x_2+x_1x_3x_4$, then $A_{23}=1$ and $A_{12}=A_{13}=A_{14}=A_{24}=A_{34}=0$. It indicates $\#\Gamma=\#\{2,3\}=2$ and $F(x_1,x_2,x_3,x_4)$ satisfies
$$F(x_1,x_2,x_3,x_4)|_{x_1=\operatorname{Tr} (b)+1,x_4=\operatorname{Tr} (\xi^{16} b)}
=(\operatorname{Tr} (b)+1)x_2+(\operatorname{Tr} (b)+1)\operatorname{Tr} (\xi^{16} b)x_3.$$
Magma shows that
$$f(x)=\operatorname{Tr}_1^3(x^{9})+\operatorname{Tr}(x)\operatorname{Tr}(\xi x)
+\operatorname{Tr}(x)\operatorname{Tr}(\xi^4x)\operatorname{Tr}(\xi^2x)$$
is bent over $\mathbb{F}_{2^6}$ and its Walsh transform at $b\in\mathbb{F}_{2^6}$ is
$$\widehat{f}(b)=2^{3}(-1)^{\operatorname{Tr}_1^3 (b^{9} )+1
+(\operatorname{Tr} (\xi^8 b)+(\operatorname{Tr} (b)+1)\operatorname{Tr} (\xi^{16} b))
(\operatorname{Tr} (\xi^{32} b)+\operatorname{Tr} (b)+1)+\operatorname{Tr} (\xi^8 b)\operatorname{Tr} (\xi^{32} b)},$$
which is consistent with our result in Theorem \ref{p=2.F1}.
\end{example}

Case II: $\#\Gamma<2\ell$, where $2\leq2\ell\leq \tau$.

If this case happens, then $i_{s_1}=i_{s_2}$ or $i_{s_1}=j_{s_2}$ for at least two integers $1\leq s_1,\,s_2\leq\ell$. Here we only discuss the case $A_{i_1j_s}=1$ for $s=1,\cdots,\ell_1$ and $A_{ij}=0$ for any other $i\ne j$ since other cases can be considered in the same manner, where $1\leq i_1< j_1< \cdots< j_{\ell_1}\leq \tau$. For this case, \eqref{f^b.p=2} becomes
$$\widehat{f}(b)= 2^{n/2-\ell_1-1}(-1)^{\widetilde{g}(b)}T,$$
where
$$\begin{aligned}
T&=\sum_{i,j\in\Gamma} \sum_{x_{i}=0}^{1}\sum_{t_{j}=0}^{1}
(-1)^{F(x_1,\cdots,x_{\tau})|_{x_i=h_i, i\not\in \Gamma}
+ (\sum\limits_{s=1}^{\ell_1}t_{j_s}+h_{i_1}+x_{i_1} )t_{i_1}+\sum\limits_{s=1}^{\ell_1} (h_{j_s}+x_{j_s})t_{j_s}}\\
&=2\sum_{i\in\Gamma} \sum_{x_{i}=0}^{1}\sum_{s=2}^{\ell_1}\sum_{t_{j_s}=0}^{1}
(-1)^{F(x_1,\cdots,x_{\tau})|_{x_i=h_i, i\not\in \Gamma}
+(h_{j_1}+x_{j_1}) (\sum\limits_{s=2}^{\ell_1}t_{j_s}+h_{i_1}+x_{i_1} ) +\sum\limits_{s=2}^{\ell_1} (h_{j_s}+x_{j_s})t_{j_s}}\\
&=2\sum_{i\in\Gamma} \sum_{x_{i}=0}^{1}\sum_{s=2}^{\ell_1}\sum_{t_{j_s}=0}^{1}
(-1)^{F(x_1,\cdots,x_{\tau})|_{x_i=h_i, i\not\in \Gamma}
+\sum\limits_{s=2}^{\ell_1}(h_{j_s}+x_{j_s}+h_{j_1}+x_{j_1})t_{j_s}+ (h_{i_1}+x_{i_1})(h_{j_1}+x_{j_1})}\\
&=2^{\ell_1}\sum_{x_{i_1}=0}^{1}\sum_{x_{j_1}=0}^{1}
(-1)^{F(x_1,\cdots,x_{\tau})|_{x_i=h_i, i\not\in \Gamma, x_{j_s}=x_{j_1}+h_{j_1}+h_{j_s}, 2\leq s\leq \ell_1}
+ (h_{i_1}+x_{i_1})(h_{j_1}+x_{j_1})}.
\end{aligned}$$

Then, another class of bent functions can be obtained as follows.

\begin{thm}\label{p=2.F2}
Let $n=2m$, $u_{1}, u_{2}, \cdots, u_{\tau}$ be pairwise distinct elements in $\mathbb{F}_{2^{n}}^{*}$, where $2\leq\tau\leq m$. Let $g(x)$ be a bent function over $\mathbb{F}_{2^n}$ whose dual satisfies \eqref{g-Aij} with $A_{i_1 j_s}=1$ for $1\leq i_1< j_1< \cdots< j_{\ell_1}\leq \tau$, $s=1,\cdots,\ell_1$ and $A_{ij}=0$ for other $1\leq i< j\leq \tau$. If $F(x_1,\cdots,x_{\tau})$ satisfies
$$F(x_1,\cdots,x_{\tau})|_{x_i=h_i, i\not\in \Gamma, x_{j_s}=x_{j_1}+h_{j_1}+h_{j_s}, 2\leq s\leq \ell_1}=
F_{i_1} x_{i_1}+F_{j_1} x_{j_1}+F_0$$
for some $F_{i_1}, F_{j_1}, F_0\in \mathbb{F}_2$, where $h_i=g_i(b)+A_{ii}$, $\Gamma$ is defined in \eqref{T}. Then $f(x)$ defined by \eqref{f=g+F} is bent and its Walsh transform at $b\in\mathbb{F}_{2^n}$ is
$$\widehat{f}(b)=2^{n/2}(-1)^{\widetilde{g}(b)+(h_{i_1}+F_{j_1})(h_{j_1}+F_{i_1})
+h_{i_1}h_{j_1}+F_0}.$$
\end{thm}

\begin{example}
Let $n=2m=8$, $\tau=4$, $\xi$ be a primitive element of $\in \mathbb{F}_{2^8}^{*}$ and $g(x)=\operatorname{Tr}_1^4(\xi^{17}x^{17})$. From \cite{M} we know that the dual of $g(x)$ satisfies \eqref{g-Aij} with $A_{ii}=\operatorname{Tr}_1^3 (\xi^{238}u_i^{17} )$, $A_{ij}=\operatorname{Tr} (\xi^{238}u_iu_j^{16} )$ for $i<j$ and $g_i(x)=\operatorname{Tr} (\xi^{238}u_i^{16}x )$. Taking $u_1=\xi$, $u_2=\xi^6$, $u_3=\xi^{11}$, $u_4=\xi^{20}$, and $F(x_1,x_2,x_3,x_4)=x_1x_4+x_2x_3x_4$, then $A_{12}=A_{13}=1$ and $A_{14}=A_{23}=A_{24}=A_{34}=0$.
It indicates $\Gamma=\{1,2,3\}$ and $F(x_1,x_2,x_3,x_4)$ satisfies
$$F(x_1,x_2,x_3,x_4)|_{x_4=\operatorname{Tr} (\xi^{48}b),x_3=x_2+\operatorname{Tr} ((\xi^{79}+\xi^{159})b)}=\operatorname{Tr} (\xi^{48}b)x_1+
\operatorname{Tr}(\xi^{48}b)(\operatorname{Tr}((\xi^{79}+\xi^{159})b)+1)x_2.$$
Magma shows that
$$f(x)=\operatorname{Tr}_1^4(\xi^{17}x^{17})+\operatorname{Tr}(\xi x)\operatorname{Tr}(\xi^{20}x)
+\operatorname{Tr}(\xi^{6}x)\operatorname{Tr}(\xi^{11}x)\operatorname{Tr}(\xi^{20}x)$$
is bent over $\mathbb{F}_{2^8}$ and its Walsh transform at $b\in\mathbb{F}_{2^8}$ is
$$\widehat{f}(b)=2^{4}(-1)^{\operatorname{Tr}_1^4(\xi^{238}b^{17})+1
+(\operatorname{Tr} (\xi^{254}b)+\operatorname{Tr}(\xi^{48}b)(\operatorname{Tr}((\xi^{79}+\xi^{159})b)+1))
(\operatorname{Tr} (\xi^{79} b)+\operatorname{Tr} (\xi^{48}b))
+\operatorname{Tr} (\xi^{254}b)\operatorname{Tr} (\xi^{79} b)},$$
which is consistent with our result in Theorem \ref{p=2.F2}.
\end{example}

\begin{remark}
For an odd prime $p$, when $A_{ii}=0$ for all $1\leq i\leq \tau$, similar results to Theorem \ref{p=2.F1} and Theorem \ref{p=2.F2} can be obtained through the above discussion.
\end{remark}
\subsection{Bent functions of the form \eqref{f=g+F} for odd $p$}\label{Walsh.g bent.odd p}
In this subsection, let $p$ be an odd prime and $f(x)$ be of the form \eqref{f=g+F}, where $g(x)$ is a weakly regular bent function whose dual satisfies \eqref{g-Aij}.

Since $g(x)$ is a weakly regular bent function, suppose that $\widehat{g}(x)=\mu^{-1}p^{n/2}\omega^{\tilde{g}(x)}$, then by \eqref{g-Aij}, \eqref{F1} and Theorem \ref{thm.wf.f=g+F}, one gets
\begin{equation}\label{f^b}
\widehat{f}(b)=\mu^{-1} p^{n/2-\tau}\sum_{i,j=1}^{\tau} \sum_{x_{i}=0}^{p-1}\sum_{t_{j}=0}^{p-1}
\omega^{F(x_1,\cdots,x_{\tau})+\widetilde{g}(b)+\sum_{1\leq i\leq j\leq \tau}A_{ij}t_it_j+\sum_{i=1}^{\tau} (g_i(b)-x_i )t_i},
\end{equation}
where $g_{i}(x)$ is a function from $\mathbb{F}_{p^{n}}$ to $\mathbb{F}_{p}$ for each $1\leq i \leq \tau$, $A_{ij}\in\mathbb{F}_{p}$ for $1\leq i\leq j\leq \tau$.

In particular, if one takes $\tau=2$ and $F(x_1,\,x_2)=x_1x_2$, then Theorem \ref{thm.wf.f=g+F} yields
$$\widehat{f}(b)=\mu^{-1} p^{n/2-\tau+1} \sum_{x_{1}=0}^{p-1}\sum_{t_{1}=0}^{p-1}
\omega^{-x_1t_1+\widetilde{g}(b)+A_{11}t_1^2+A_{12}x_1t_1+A_{22}x_1^2+g_1(b)t_1+g_2(b)x_1}.$$

As a consequence, the following result can be verified by Lemma \ref{expH.k=2}.

\begin{prop}\label{cor.bent g(x).k=2}
Let $p$ be an odd prime and $g(x)$ be a weakly regular bent function over $\mathbb{F}_{p^n}$ whose dual satisfies
$$\widetilde{g}(x-t_1u_1-t_2u_2)=\widetilde{g}(x)+A_{11}t_1^2+A_{22}t_2^2+A_{12}t_1t_2+g_1(x)t_1+g_2(x)t_2$$
for all $x\in\mathbb{F}_{p^n}$ and $t_1,\,t_2\in\mathbb{F}_{p}$, where $g_1(x)$ and $g_2(x)$ are functions from $\mathbb{F}_{p^{n}}$ to $\mathbb{F}_{p}$. Let $u_1, u_2\in\mathbb{F}_{p^{n}}^{*}$, then $f(x)=g(x)+\operatorname{Tr}(u_1x)\operatorname{Tr}(u_2x)$ is a weakly regular bent function if and only if $(A_{12}-1)^2-4A_{11}A_{22}\neq 0$. Moreover, for $b \in\mathbb{F}_{p^n}$, the Walsh transform of $f(x)$ is
$$\widehat{f}(b)=\mu^{-1}\eta ((A_{12}-1)^2-4A_{11}A_{22} )p^{n/2}
\omega^{\widetilde{g}(b)+\frac{A_{22}g_1(b)^2+A_{11}g_2(b)^2-(A_{12}-1)g_1(b)g_2(b)}{(A_{12}-1)^2-4A_{11}A_{22}}
}.$$
\end{prop}

\begin{example}
Let $p=5$, $n=2$, $u_{1}, u_{2}\in\mathbb{F}_{5^2}^{*}$ and $g(x)=\operatorname{Tr}(ax^{2})$ with $a\in \mathbb{F}_{5^2}^{*}$. It is known in \cite{HK} that $\widehat{g}(b)=-\eta(a)5\omega^{\operatorname{Tr}(\frac{b^{2}}{a})}$ for $b\in\mathbb{F}_{5^2}$. Then   $\widetilde{g}(x-t_1u_1-t_2u_2)$ is equal to
$$\operatorname{Tr}(\frac{x^{2}}{a})
+\operatorname{Tr}(\frac{u_1^{2}}{a})t_1^2+\operatorname{Tr}(\frac{u_2^{2}}{a})t_2^2+
2\operatorname{Tr}(\frac{u_1u_2}{a})t_1t_2-2\operatorname{Tr}(\frac{u_1x}{a})t_1
-2\operatorname{Tr}(\frac{u_2x}{a})t_2.$$
Magma experiments show that $f(x)=\operatorname{Tr}(ax^{2})+\operatorname{Tr}(u_1x)\operatorname{Tr}(u_2x)$ is bent if and only if $\Delta=(2\operatorname{Tr}(\frac{u_1u_2}{a})-1)^2+ \operatorname{Tr}(\frac{u_1^2}{a})\operatorname{Tr}(\frac{u_2^2}{a})\neq0$.   Moreover, for $b \in\mathbb{F}_{5^2}$, the Walsh transform of $f(x)$ is
$$\widehat{f}(b)=-\eta(a\Delta)5
\omega^{\operatorname{Tr}(\frac{b^{2}}{a})
-\frac{\operatorname{Tr}(\frac{u_2^2}{a})\operatorname{Tr}(\frac{u_1b}{a})^2
+\operatorname{Tr}(\frac{u_1^2}{a})\operatorname{Tr}(\frac{u_2b}{a})^2
-(2\operatorname{Tr}(\frac{u_1u_2}{a})-1)
\operatorname{Tr}(\frac{u_1b}{a})\operatorname{Tr}(\frac{u_2b}{a})}
{(2\operatorname{Tr}(\frac{u_1u_2}{a})-1)^2+ \operatorname{Tr}(\frac{u_1^2}{a})\operatorname{Tr}(\frac{u_2^2}{a})}
}.$$
\end{example}

For a general $F(x_1,\cdots,x_{\tau})$, we assume that the nonzero elements in $\{A_{ij}:1\leq i\leq j\leq \tau\}$ are $A_{i_1j_1}, \cdots, A_{i_{\ell}j_{\ell}}$, where $\ell$ is an integer, $1\leq i_1\leq \cdots\leq i_{\ell}\leq \tau$ and $1< j_1\leq \cdots\leq j_{\ell}\leq \tau$. Let $\Gamma$ be given as in \eqref{T}, it is extremely difficult to calculate the value of $\widehat{f}(b)$ when $\deg(F(x_1,\cdots,x_{\tau})|_{x_i=g_i(b), i\not\in\Gamma})>2$. Thus, we focus on the case $\deg(F(x_1,\cdots,x_{\tau})|_{x_i=g_i(b),i\not\in\Gamma})\leq2$,
and then we can discuss \eqref{f^b} as follows:

Case I: $\#\Gamma\leq2$.

Assume that $\tau_1\ne\tau_2$ and $\Gamma\subset\{\tau_1,\tau_2\}$, i.e., $A_{\tau_1\tau_1},A_{\tau_1\tau_2},A_{\tau_2\tau_2}\in \mathbb{F}_{p}$ and $A_{ij}=0$ otherwise. Then  \eqref{f^b} is reduced to
$$\widehat{f}(b)=\mu^{-1} p^{n/2-2}\omega^{\widetilde{g}(b)}M$$
by using the fact that
$\sum_{t_{i}=0}^{p-1} \omega^{(g_{i}(b)-x_{i}) t_{i}}=p$ if $x_{i}=g_{i}(b)$ and $0$ otherwise for $i\not\in\{\tau_1,\tau_2\}$, where
$$M=
\sum_{x_{\tau_1}=0}^{p-1}\sum_{x_{\tau_2}=0}^{p-1}
\omega^{F'}
\sum_{t_{\tau_1}=0}^{p-1}\sum_{t_{\tau_2}=0}^{p-1}\omega^{
A_{\tau_1\tau_1}t_{\tau_1}^2+A_{\tau_2\tau_2}t_{\tau_2}^2+A_{\tau_1\tau_2}t_{\tau_1}t_{\tau_2}
+ (g_{\tau_1}(b)-x_{\tau_1} )t_{\tau_1}+ (g_{\tau_2}(b)-x_{\tau_2} )t_{\tau_2}}
$$
with $F'=F(x_1,\cdots,x_{\tau})|_{x_i=g_{i}(b), i\not\in\{\tau_1,\tau_2\}}$. Then, according to Lemma \ref{expH.k=2}, we can calculate \eqref{f^b} as the following four cases:

(1): $A_{\tau_1\tau_1}=A_{\tau_1\tau_2}=A_{\tau_2\tau_2}=0$. For this case, by Lemma \ref{expH.k=2}, it can be readily verified that
$$\widehat{f}(b)=\mu^{-1} p^{n/2}\omega^{\widetilde{g}(b)+F (g_1(b),g_2(b),\cdots,g_{\tau}(b) )}.$$

(2): $A_{\tau_1\tau_1}=A_{\tau_1\tau_2}=0$ and $A_{\tau_2\tau_2}\neq 0$. If this case occurs, then again by Lemma \ref{expH.k=2} one can claim that $g_{\tau_1}(b)-x_{\tau_1}=0$ and
$$\widehat{f}(b)=\mu^{-1} p^{n/2-1}\sqrt{p^*}\eta(A_{\tau_2\tau_2})\omega^{\widetilde{g}(b)}\sum_{x_{\tau_2}=0}^{p-1}
\omega^{F(x_1,\cdots,x_{\tau})|_{x_i=g_{i}(b), i\ne\tau_2}-\frac{ (g_{\tau_2}(b)-x_{\tau_2} )^2}{4A_{\tau_2\tau_2}}}.$$
Under this case, if $\deg(F(x_1,\cdots,x_{\tau})|_{x_i=g_i(b), i\ne\tau_2})\leq2$, then $F(x_1,\cdots,x_{\tau})$ satisfies
\begin{equation}\label{F{1}}
F(x_1,\cdots,x_{\tau})|_{x_i=g_{i}(b), i\ne\tau_2}=a_{2}x_{\tau_2}^2+F_{2}(b)x_{\tau_2}+F_0(b)
\end{equation}
for some $F_{2}(x),F_0(x)\in\mathbb{F}_p[x]$ and $a_{2}\in\mathbb{F}_p$. Consequently, one obtains
$$\widehat{f}(b)=\mu^{-1} p^{n/2-1}\sqrt{p^*}\eta(A_{\tau_2\tau_2})\omega^{\widetilde{g}(b)}\sum_{x_{\tau_2}=0}^{p-1}
\omega^{ (a_{2}-\frac{1}{4A_{\tau_2\tau_2}} )x_{\tau_2}^2
+ (F_2(b)+\frac{g_{\tau_2}(b)}{2A_{\tau_2\tau_2}} )x_{\tau_2}+F_0(b)-\frac{g_{\tau_2}(b)^2}{4A_{\tau_2\tau_2}}
}.$$
Then we can deduce that $f(x)$ is bent only if $a_2-\frac{1}{4A_{\tau_2\tau_2}}\neq 0$ and
$$\widehat{f}(b)=\mu^{-1}p^{n/2}\eta(1-4a_2A_{\tau_2\tau_2})\omega^{\widetilde{g}(b)
+\frac{A_{\tau_2\tau_2}}{1-4a_{2}A_{\tau_2\tau_2}} (F_2(b)+\frac{g_{\tau_2}(b)}{2A_{\tau_2\tau_2}} )^2
+F_0(b)-\frac{g_{\tau_2}(b)^2}{4A_{\tau_2\tau_2}}}
$$
from Lemma \ref{char sum(ax^2+bx+c)}.

(3): $A_{\tau_1\tau_1}\neq 0$ and $A_{\tau_1 \tau_2}^2=4A_{\tau_1 \tau_1}A_{\tau_2 \tau_2}$. In this case, by Lemma \ref{expH.k=2}, one has  $$g_{\tau_2}(b)-x_{\tau_2}=\frac{A_{\tau_1 \tau_2} (g_{\tau_1}(b)-x_{\tau_1} )}{2A_{\tau_1 \tau_1}},$$
i.e., $x_{\tau_2}=\frac{A_{\tau_1 \tau_2}}{2A_{\tau_1 \tau_1}}x_{\tau_1}+\varphi(b)$ with $\varphi(b)=g_{\tau_2}(b)-\frac{A_{\tau_1 \tau_2}}{2A_{\tau_1 \tau_1}}g_{\tau_1}(b)$. Then $M$ becomes
$$M=p\sqrt{p^*}\eta(A_{\tau_1\tau_1})\sum_{x_{\tau_1}=0}^{p-1}
\omega^{F(x_1,\cdots,x_{\tau})|_{x_i=g_{i}(b), i\not\in\{\tau_1,\tau_2\}, x_{\tau_2}=A_{\tau_1 \tau_2}x_{\tau_1}/(2A_{\tau_1 \tau_1})+\varphi(b)}-\frac{ (g_{\tau_1}(b)-x_{\tau_1} )^2}{4A_{\tau_1 \tau_1}}}.
$$
 Note that $\deg(F(x_1,\cdots,x_{\tau})|_{x_i=g_i(b), i\not\in\{\tau_1,\tau_2\}})\leq2$ yields $F(x_1,\cdots,x_{\tau})$ satisfies
\begin{equation}\label{F{12}}
F(x_1,\cdots,x_{\tau})|_{x_i=g_{i}(b), i\not\in\{\tau_1,\tau_2\}}
=a_{11}x_{\tau_1}^2+a_{22}x_{\tau_2}^2+a_{12}x_{\tau_1}x_{\tau_2}+F_{1}(b)x_{\tau_1}+F_{2}(b)x_{\tau_2}+F_0(b)
\end{equation}
for some $F_{1}(x),F_{2}(x),F_0(x)\in\mathbb{F}_p[x]$ and $a_{11},a_{22},a_{12}\in\mathbb{F}_p$, then $M$ turns into
$$M=p\sqrt{p^*}\eta(A_{\tau_1\tau_1})\sum\nolimits_{x_{\tau_1}=0}^{p-1}\omega^{\alpha_1x_{\tau_1}^2
+\alpha_2x_{\tau_1}+\alpha_3},$$
where $\alpha_{i}$ are given by
\begin{eqnarray}\label{ai}
  \alpha_1&=&a_{11}+ (\frac{A_{\tau_1 \tau_2}}{2A_{\tau_1 \tau_1}} )^2a_{22}+
  \frac{2A_{\tau_1\tau_2}a_{12}-1}{4A_{\tau_1\tau_1}},\nonumber \\
  \alpha_2 &=&  (\frac{A_{\tau_1 \tau_2}}{A_{\tau_1 \tau_1}}a_{22}+a_{12} )\varphi(b)+F_1(b)+\frac{A_{\tau_1\tau_2}F_{2}(b)+g_{\tau_1}(b)}{2A_{\tau_1\tau_1}}, \\
  \alpha_3&=& a_{22}\varphi(b)^2+F_2(b)\varphi(b)+F_0(b)-\frac{g_{\tau_1}(b)^2}{4A_{\tau_1\tau_1}}. \nonumber
\end{eqnarray}
Then it can be verified that $f(x)$ is bent only if $\alpha_1\neq0$ and in this case we have
$$\widehat{f}(b)=\mu^{-1} p^{n/2}\eta(-A_{\tau_1 \tau_1}\alpha_1)\omega^{\widetilde{g}(b)+\alpha_3-\alpha_2^2/(4\alpha_1)}.$$
(4): $A_{\tau_1\tau_2}^{2}-4 A_{\tau_1\tau_1}A_{\tau_2\tau_2}\neq 0$. When $F(x_1,\cdots,x_{\tau})$ satisfies \eqref{F{12}}, similar to the above case, a straightforward computation gives
$$M=p\eta (\Delta)\sum_{x_{\tau_1}=0}^{p-1} \sum_{x_{\tau_2}=0}^{p-1} \omega^{B_{1}x_{\tau_1}^2+B_{2}x_{\tau_2}^2+B_{3}x_{\tau_1}x_{\tau_2}+\beta_{1}x_{\tau_1}
+\beta_{2}x_{\tau_2}+\beta_{3}},$$
where $\Delta=A_{\tau_1 \tau_2}^2-4A_{\tau_1 \tau_1}A_{\tau_2 \tau_2}$, $B_{i}$ are given by
\begin{eqnarray}\label{Bi}
 B_1 = a_{11}+\frac{A_{\tau_2\tau_2}}{\Delta},
 B_2 = a_{22}+\frac{A_{\tau_1\tau_1}}{\Delta},
 B_3 = a_{12}-\frac{A_{\tau_1\tau_2}}{\Delta}
\end{eqnarray}
and $\beta_{i}$ are given by
\begin{eqnarray}\label{betai}
  \beta_1&=&F_{1}(b)+\frac{A_{\tau_1\tau_2}g_{\tau_2}(b)-2A_{\tau_2\tau_2}g_{\tau_1}(b)}{\Delta},\nonumber \\
  \beta_2 &=& F_{2}(b)+\frac{A_{\tau_1\tau_2}g_{\tau_1}(b)-2A_{\tau_1\tau_1}g_{\tau_2}(b)}{\Delta}, \\
  \beta_3&=& F_0(b)+\frac{A_{\tau_2\tau_2}g_{\tau_1}(b)^2+A_{\tau_1\tau_1}g_{\tau_2}(b)^2
-A_{\tau_1\tau_2}g_{\tau_1}(b)g_{\tau_2}(b)}{\Delta}. \nonumber
\end{eqnarray}
Then by Lemma \ref{expH.k=2} one can conclude that $f(x)$ is bent only if $B_{3}^2-4B_{1}B_{2}\neq 0$ and consequently
$$\widehat{f}(b)=\mu^{-1}p^{n/2}\eta (\Delta )\eta (B_{3}^2-4B_{1}B_{2} )
\omega^{\widetilde{g}(b)+\frac{B_{1}\beta_{2}^2+B_{2}\beta_{1}^2-B_{3}\beta_{1}\beta_{2}}
{B_{3}^2-4B_{1}B_{2}}+\beta_{3}}.$$

Based on the above discussions, we can arrive at the following result.

\begin{thm}\label{thm.bent g(x)}
Let $p$ be an odd prime, $u_{1}, u_{2}, \cdots, u_{\tau}\in\mathbb{F}_{p^{n}}^{*}$ with $\tau\geq2$ and $f(x)$ be defined by \eqref{f=g+F}, where $g(x)$ is a weakly regular bent function over $\mathbb{F}_{p^n}$ whose dual satisfies \eqref{g-Aij} and $F(x_1,\cdots,x_n)$ is a reduced polynomial in $\mathbb{F}_p[x_1,\cdots,x_n]$.
\\
(1) If $A_{ij}=0$ for all $1\leq i\leq j\leq \tau$, then $f(x)$ is weakly regular bent and its Walsh transform is
$$\widehat{f}(b)=\mu^{-1}p^{n/2}\omega^{\widetilde{g}(b)+F (g_1(b),g_2(b),\cdots,g_{\tau}(b) )}.$$
(2) If there exists some $1\leq \tau_2\leq \tau$ such that $A_{\tau_2\tau_2}\neq 0$, $A_{ij}=0$ otherwise, $F(x_1,\cdots,x_{\tau})$ satisfies \eqref{F{1}} and $a_2-\frac{1}{4A_{\tau_2\tau_2}}\neq 0$, then $f(x)$ is weakly regular bent and its Walsh transform is
$$\widehat{f}(b)=\mu^{-1}p^{n/2}\eta(1-4a_2A_{\tau_2\tau_2})\omega^{\widetilde{g}(b)
+\frac{A_{\tau_2\tau_2}}{1-4a_{2}A_{\tau_2\tau_2}} (F_2(b)+\frac{g_{\tau_2}(b)}{2A_{\tau_2\tau_2}} )^2
+F_0(b)-\frac{g_{\tau_2}(b)^2}{4A_{\tau_2\tau_2}}}.
$$
(3) If there exist some $1\leq \tau_1<\tau_2\leq \tau$ such that $A_{\tau_1\tau_1}\neq 0$, $A_{\tau_1\tau_2}^2-4A_{\tau_1\tau_1}A_{\tau_2\tau_2}=0$, $A_{ij}=0$ otherwise, $F(x_1,\cdots,x_{\tau})$ satisfies \eqref{F{12}} and $\alpha_1\neq 0$, then $f(x)$ is weakly regular bent and its Walsh transform is
$$\widehat{f}(b)=\mu^{-1} p^{n/2}\eta(-A_{\tau_1\tau_1}\alpha_1)\omega^{\widetilde{g}(b)+\alpha_3-\alpha_2^2/(4\alpha_1)},$$
where $\alpha_{i}$ are given by \eqref{ai}.\\
(4) If there exist some $1\leq \tau_1< \tau_2\leq \tau$ such that $A_{\tau_1 \tau_2}^2-4A_{\tau_1 \tau_1}A_{\tau_2 \tau_2}\neq 0$, $A_{ij}=0$ otherwise, $F(x_1,\cdots,x_{\tau})$ satisfies \eqref{F{12}} and $B_{3}^2-4B_{1}B_{2}\neq 0$, then $f(x)$ is weakly regular bent and its Walsh transform is
$$\widehat{f}(b)=\mu^{-1} p^{n/2}\eta (A_{\tau_1\tau_2}^{2}-4 A_{\tau_1\tau_1}A_{\tau_2\tau_2} )\eta (4B_{1}B_{2}-B_{3}^2 )\omega^{\widetilde{g}(b)
+\frac{B_{1}\beta_{2}^2+B_{2}\beta_{1}^2-B_{3}\beta_{1}\beta_{2}}{B_{3}^2-4B_{1}B_{2}}+\beta_{3}},$$
where $B_{i}$, $\beta_{i}$ are given by  \eqref{Bi} and \eqref{betai} respectively.
\end{thm}

\begin{remark}
Notice that the item $A_{ii}t_i^2$ will be reduced to $A_{ii}t_i$ if $p=2$ and for this case the result in (1) of Theorem \ref{thm.bent g(x)} still holds, which has been studied in \cite{T}. For an odd prime $p$, the result in (1) of Theorem \ref{thm.bent g(x)} has been obtained in \cite[Theorem 3.2]{QTZF}.
\end{remark}

\begin{example}
Let $p=3$, $n=4$, $\tau=3$, $\xi$ be a primitive element of $\mathbb{F}_{3^4}^{*}$ and $g(x)=\operatorname{Tr}(x^{2})$. From \cite{HK} we know that
$\widehat{g}(b)=-3^{2}\omega^{-\operatorname{Tr}(b^{2})}$ for $b\in\mathbb{F}_{3^4}$ and $\widetilde{g}(x)$ satisfies \eqref{g-Aij} with $A_{ii}=-\operatorname{Tr}(u_i^{2})$, $A_{ij}=\operatorname{Tr}(u_iu_j)$ if $i< j$ and $g_i(x)=-\operatorname{Tr}(u_ix)$. Taking $u_3=\xi^{53}$ and $F(x_1,x_2,x_3)=x_1x_3^2+x_2x_3$, then $A_{33}=0$.
Then Magma shows
$$f(x)=\operatorname{Tr}(x^{2})+\operatorname{Tr}(u_1x)\operatorname{Tr}(\xi^{53}x)^2
+\operatorname{Tr}(u_2x)\operatorname{Tr}(\xi^{53}x)
$$
is a weakly regular bent function if one of the following conditions satisfies\\
(1) $u_1=u_2=\xi^{13}$. In this case $A_{11}=A_{22}=A_{12}=A_{13}=A_{23}=0$ and the Walsh transform of $f(x)$ is
$$\widehat{f}(b)=-3^2\omega^{-\operatorname{Tr}(b^{2})
-\operatorname{Tr}(\xi^{13}b)\operatorname{Tr}(\xi^{53}b)^2
+\operatorname{Tr}(\xi^{13}b)\operatorname{Tr}(\xi^{53}b)}.$$
(2) $u_1=\xi^{13}$ and $u_2=\xi^2$. This case gives $A_{22}=1$, $A_{11}=A_{12}=A_{13}=A_{23}=0$ and
$$F(x_1,x_2,x_{3})|_{x_1=-\operatorname{Tr}(\xi^{13}b),x_3=-\operatorname{Tr}(\xi^{53}b)}
=-\operatorname{Tr}(\xi^{53}b)x_2-\operatorname{Tr}(\xi^{13}b)\operatorname{Tr}(\xi^{53}b)^2,$$
which indicates $a_2=0$, $a_2+1/\operatorname{Tr}(\xi^4)=1\neq 0$, and then the Walsh transform of $f(x)$ is
$$\widehat{f}(b)=-3^2\omega^{-\operatorname{Tr}(b^{2})
+\operatorname{Tr}((\xi^2-\xi^{53})b)^2-\operatorname{Tr}(\xi^{13}b)
\operatorname{Tr}(\xi^{53}b)^2
-\operatorname{Tr}(\xi^2b)^2}.
$$
(3) $u_1=\xi^2$ and $u_2=\xi^7$. For this case, we have $A_{11}=A_{22}=A_{12}=1$ and $A_{13}=A_{23}=0$, which implies that $A_{12}^2-A_{11}A_{22}=0$ and
$$F(x_1,x_2,x_{3})|_{x_3=-\operatorname{Tr}(\xi^{53}b)}=\operatorname{Tr}(\xi^{53}b)^2x_1
-\operatorname{Tr}(\xi^{53}b)x_2.$$
Then
$\alpha_1=-1$, $\alpha_2 = \operatorname{Tr}(\xi^{53}b)^2+\operatorname{Tr}((\xi^{53}+\xi^{2})b)$, $\alpha_3= \operatorname{Tr}(\xi^{53}b)\operatorname{Tr}((\xi^{2}+\xi^{7})b)-\operatorname{Tr}(\xi^{2}b)^2$
and
$$\widehat{f}(b)=-3^2\omega^{-\operatorname{Tr}(b^{2})
+\operatorname{Tr}(\xi^{53}b)\operatorname{Tr}((\xi^{2}+\xi^{7})b)-\operatorname{Tr}(\xi^{2}b)^2
+(\operatorname{Tr}(\xi^{53}b)^2+\operatorname{Tr}((\xi^{53}+\xi^{2})b))^2}.$$
(4) $u_1=\xi^2$ and $u_2=\xi^9$. In this case we have $A_{11}=A_{12}=1$, $A_{22}=-1$, $A_{13}=A_{23}=0$, $\Delta=A_{12}^2-A_{11}A_{22}=-1$,
$B_1=1$, $B_2=-1$,  $B_3=1$, $\beta_1=\operatorname{Tr}(\xi^{53}b)^2+\operatorname{Tr}((\xi^{9}-\xi^{2})b)
$, $\beta_2 =\operatorname{Tr}((\xi^{9}+\xi^{2}-\xi^{53})b)$ and $\beta_3=\operatorname{Tr}(\xi^{2}b)^2-\operatorname{Tr}(\xi^{9}b)^2
+\operatorname{Tr}(\xi^{2}b)\operatorname{Tr}(\xi^{9}b)$.
Then $B_3^2-B_1B_2=-1\neq 0$ and the Walsh transform of $f(x)$ is
$$\widehat{f}(b)=-3^2\omega^{-\operatorname{Tr}(b^{2})+\beta_{1}^2-\beta_{2}^2
+\beta_{1}\beta_{2}+\beta_{3}}.$$
Computer experiments are consistent with our results in Theorem \ref{thm.bent g(x)}.
\end{example}

The algebraic degree of the weakly regular bent functions in Theorem \ref{thm.bent g(x)} can be determined for some $u_i$
by Lemma \ref{algebraic degree}.

\begin{prop}\label{algebraic degree.f=g+F}
Let $u_{1}, u_{2}, \cdots, u_{\tau}\in\mathbb{F}_{p^{n}}^{*}$ be linearly independent over $\mathbb{F}_{p},$ where $\tau$ is an integer with $2\leq\tau\leq n.$ Let $F(x_1,\cdots,x_{\tau})$ be a reduced polynomial in $\mathbb{F}_p[x_1,\cdots,x_{\tau}]$ with algebraic degree $d$. Then the algebraic degree of the weakly regular bent function constructed in Theorem \ref{thm.bent g(x)} is $\max \{d,\deg(g(x))\}$.
\end{prop}

The algebraic degree of $f(x)$ in Theorem \ref{thm.bent g(x)} can achieve the upper bound in Lemma \ref{algebraic degree. limit}. Let $u_{1}, u_{2}, \cdots, u_{\tau}\in\mathbb{F}_{p^{n}}^{*}$ be linearly independent over $\mathbb{F}_{p}$, where $2\leq\tau\leq n$. Let $f(x)$ be the weakly regular bent function generated by  Theorem \ref{thm.bent g(x)}, then the algebraic degree of $f(x)$ is $\frac{(p-1)n}{2}$ if one takes $p\geq 3$ and $F(x_1,\cdots,x_{\tau})=\prod_{i=1}^{\tau}x_i^{e_i}$ with $\sum_{i=1}^{\tau} e_i=\frac{(p-1)n}{2}$, where $2\leq\tau\leq n$, $0\leq e_i\leq p-1$.

Case II: $\#\Gamma>2$.

In this case, in order to guarantee $\deg(F(x_1,\cdots,x_{\tau})|_{x_i=g_i(b), i\not\in\Gamma})\leq2$ for a general $\Gamma$, we consider the case $\deg(F(x_1,\cdots,x_{\tau}))=2$.  Without loss of generality, assume that
\begin{equation}\label{f=g+x_ix_j}
f(x)=g(x)+\sum\nolimits_{1\leq i\leq j \leq \tau}a_{ij}\operatorname{Tr}(u_{i}x)\operatorname{Tr}(u_{j}x),
\end{equation}
where $a_{ij}\in\mathbb{F}_{p}$ are not all zero. Let $g(x)$ be a weakly regular bent function over $\mathbb{F}_{p^n}$ whose dual satisfies \eqref{g-Aij} with $A_{ij}=0$ for $i\neq j$ (the case $A_{ij}$ are not all zero when $i\neq j$ can be similarly considered), i.e.,
$$\widetilde{g}(x-\sum\nolimits_{i=1}^{\tau}u_it_i)=\widetilde{g}(x)
+\sum\nolimits_{i=1}^{\tau}A_{ii}t_i^2+\sum\nolimits_{i=1}^{\tau}g_i(x)t_i.
$$
 Then \eqref{f^b} is reduced to
\begin{equation}\label{f2}
\widehat{f}(b)=\mu^{-1} p^{n/2-\tau}\sum_{i,j=1}^{\tau} \sum_{x_{i}=0}^{p-1}\sum_{t_{j}=0}^{p-1}
\omega^{\sum\nolimits_{1\leq i\leq j \leq \tau}a_{ij}x_{i}x_{j}+\widetilde{g}(b)+\sum_{i=1}^{\tau}A_{ii}t_i^2+\sum_{i=1}^{\tau}  (g_i(b)-x_i )t_i}.
\end{equation}
Observe that if $A_{ii}=0$ for some $1\leq i\leq \tau$, say $A_{11}=0$, then $\widehat{f}(b)$ equals
$$\mu^{-1} p^{\frac{n}{2}-\tau+1}\sum_{i,j=2}^{\tau} \sum_{x_{i}=0}^{p-1}\sum_{t_{j}=0}^{p-1}
\omega^{\sum\limits_{2\leq i\leq j \leq \tau}a_{ij}x_{i}x_{j}+\widetilde{g}(b)+\sum\limits_{i=2}^{\tau}A_{ii}t_i^2+\sum\limits_{i=2}^{\tau} (g_i(b)-x_i )t_i}\cdot
\omega^{\sum\limits_{j=2}^{\tau}a_{1j}g_{1}(b)x_{j}+a_{11}g_{1}(b)^2} ,$$
which implies that $f(x)$ degenerates to the case
$$f(x)=g(x)+\sum\nolimits_{2\leq i\leq j \leq \tau}a_{ij}\operatorname{Tr}(u_{i}x_{i})\operatorname{Tr}(u_{j}x_{j}).$$
Thus, we next assume that $A_{ii}\neq 0$ for all $1\leq i\leq \tau$.

For simplicity, define $\gamma_{i}^{(1)}=a_{ii}-1/(4A_{ii})$, $\varrho_{i}^{(1)}=g_i(b)/(2A_{ii})$ for $1\leq i\leq \tau$, $\gamma_{i,j}^{(1)}=a_{ij}$ for $1\leq i<j\leq \tau$ and
\begin{eqnarray}\label{Gi}
  \gamma_{i}^{(k)}&=& \gamma_{i}^{(k-1)}- (\gamma_{k-1,i}^{(k-1)} )^2 {/}
 (4\gamma_{k-1}^{(k-1)} ),\nonumber \\
 \gamma_{i,j}^{(k)} &=& \gamma_{i,j}^{(k-1)}-\gamma_{k-1,i}^{(k-1)}\gamma_{k-1,j}^{(k-1)} {/} (2\gamma_{k-1}^{(k-1)} ), \\
 \varrho_{i}^{(k)}&=&\varrho_{i}^{(k-1)}-\gamma_{k-1,i}^{(k-1)}\varrho_{k-1}^{(k-1)} {/}
 (2\gamma_{k-1}^{(k-1)} ) \nonumber
\end{eqnarray}
 for $2\leq k\leq \tau$. Then, according to \eqref{f2} and Lemma \ref{char sum(ax^2+bx+c)},  one gets
$$\widehat{f}(b)=\mu^{-1} p^{n/2-\tau}\sqrt{p^*}^{\tau}\prod\nolimits_{i=1}^{\tau}\eta(A_{ii})\omega^{\widetilde{g}(b)}N,$$
where
$$\begin{aligned}
N&=\sum_{i=1}^{\tau} \sum_{x_{i}=0}^{p-1}
\omega^{\sum_{i=1}^{\tau} (-\frac{ (g_i(b)-x_i )^2}{4A_{ii}}+a_{ii}x_i^2 )
+\sum_{i=1}^{\tau-1}\sum_{j=i+1}^{\tau}a_{ij}x_{i}x_{j}}\\
&=\omega^{-\sum_{i=1}^{\tau}\frac{g_i(b)^2}{4A_{ii}}}\sum_{i=1}^{\tau} \sum_{x_{i}=0}^{p-1}
\omega^{\sum_{i=1}^{\tau-1} (\gamma_{i}^{(1)}x_{i}^2+ (\sum_{j=i+1}^{\tau}\gamma_{i,j}^{(1)}x_{j}
+\varrho_{i}^{(1)} )x_i )+\gamma_{\tau}^{(1)}x_{\tau}^2+\varrho_{\tau}^{(1)}x_{\tau}}.
\end{aligned}$$
If $\gamma_{1}^{(1)}\neq0$, then by Lemma \ref{char sum(ax^2+bx+c)} one obtains
$$N=\omega^{-\sum_{i=1}^{\tau}\frac{g_i(b)^2}{4A_{ii}}-
 (\varrho_{1}^{(1)})^2/4\gamma_{1}^{(1)}}\sqrt{p^*}\eta (\gamma_{1}^{(1)} )\sum_{i=2}^{\tau} \sum_{x_{i}=0}^{p-1}
\omega^{\sum_{i=2}^{\tau-1} (\gamma_{i}^{(2)}x_{i}^2+ (\sum_{j=i+1}^{\tau}\gamma_{i,j}^{(2)}x_{j}+
\varrho_{i}^{(2)} )x_i )+\gamma_{\tau}^{(2)}x_{\tau}^2+\varrho_{\tau}^{(2)}x_{\tau}},$$
 Therefore, if $\gamma_{i}^{(i)}\neq 0$ for all $1\leq i\leq \tau$, one can conclude that
$$N=\omega^{-\sum_{i=1}^{\tau}\frac{g_i(b)^2}{4A_{ii}}}N^{(\tau-1)}
\sum_{x_{\tau}=0}^{p-1}
\omega^{\gamma_{\tau}^{(\tau)}x_{\tau}^2+\varrho_{\tau}^{(\tau)}x_{\tau}}
=\omega^{-\sum_{i=1}^{\tau}\frac{g_i(b)^2}{4A_{ii}}}N^{(\tau)},$$
where  $N^{(k)}$, $1\leq k\leq \tau$ is defined as
$$N^{(k)}=\sqrt{p^*}^{k}\prod\nolimits_{i=1}^{k}\eta (\gamma_{i}^{(i)} )\omega^{-\sum_{i=1}^{k}
 (\varrho_{i}^{(i)} )^2 {/} (4\gamma_{i}^{(i)} )}.$$

\begin{thm}\label{thm.bent g(x)+x_ix_j}
Let $p$ be an odd prime and $u_{1}, u_{2}, \cdots, u_{\tau}\in\mathbb{F}_{p^{n}}^{*}$ with $\tau\geq2$.
Let $g(x)$ be a weakly regular bent function over $\mathbb{F}_{p^n}$ whose dual satisfies \eqref{g-Aij} with $A_{ij}=0$ for $i\neq j$ and $A_{ii}\neq 0$ for all $1\leq i\leq \tau$. Then the function $f(x)$ of the form \eqref{f=g+x_ix_j} is a weakly regular bent function if $\gamma_{i}^{(i)}\neq 0$ for all $1\leq i\leq \tau$. Moreover,  the Walsh transform of $f(x)$ is
$$\widehat{f}(b)=\mu^{-1} p^{n/2}\prod\nolimits_{i=1}^{\tau}\eta ((-1)^{\tau}A_{ii}\gamma_{i}^{(i)} )
\omega^{\widetilde{g}(b)-\sum_{i=1}^{\tau} (\frac{g_i(b)^2}{4A_{ii}}+ (\varrho_{i}^{(i)} )^2 {/}
 (4\gamma_{i}^{(i)} ) )},$$
where $\gamma_{i}^{(k)}$, $\gamma_{i,j}^{(k)}$ and $\varrho_{i}^{(k)}$ are given by \eqref{Gi} for $1\leq k\leq \tau$.
\end{thm}

\begin{example}
Let $p=3$, $n=5$, $\tau=4$, $\xi$ be a primitive element of $\in \mathbb{F}_{3^5}^{*}$ and $g(x)=\operatorname{Tr}(x^{2})$. From \cite{HK} we know that $\widehat{g}(b)=(-3)^{5/2}\omega^{-\operatorname{Tr}(b^{2})}$ for $b\in\mathbb{F}_{3^5}$ and $\widetilde{g}(x)$ satisfies
\eqref{g-Aij} with $A_{ii}=-\operatorname{Tr}(u_i^{2})$,
$A_{ij}=\operatorname{Tr}(u_iu_j)$ if $i< j$ and $g_i(x)=-\operatorname{Tr}(u_ix)$.
Taking $u_1=\xi^2$, $u_2=\xi^5$, $u_3=\xi^4$, $u_4=\xi^{16}$, and $F(x_1,x_2,x_3,x_4)=x_1x_2+x_3x_4$, then $A_{11}=A_{33}=-1$, $A_{22}=A_{44}=1$, $A_{12}=A_{13}=A_{14}=A_{23}=A_{24}=A_{34}=0$, and from \eqref{Gi}, we have $\gamma_{1}^{(1)}=\gamma_{2}^{(2)}=\gamma_{3}^{(3)}=\gamma_{4}^{(4)}=1$,
$\varrho_{1}^{(1)}=-\operatorname{Tr}(\xi^2b)$, $\varrho_{2}^{(2)}=\operatorname{Tr}((\xi^5-\xi^2)b)$, $\varrho_{3}^{(3)}=-\operatorname{Tr}(\xi^4b)$, $\varrho_{4}^{(4)}=\operatorname{Tr}((\xi^{16}-\xi^4)b)$.
Magma shows that
$$f(x)=\operatorname{Tr}(x^{2})+\operatorname{Tr}(\xi^2x)\operatorname{Tr}(\xi^5x)
+\operatorname{Tr}(\xi^4x)\operatorname{Tr}(\xi^{16}x)$$
is bent over $\mathbb{F}_{3^5}$. Moreover, for $b \in\mathbb{F}_{3^5}$, the Walsh transform of $f(x)$ is
$$\widehat{f}(b)=(-3)^{5/2}\omega^{-\operatorname{Tr} (b^{2} )
-\operatorname{Tr} (\xi^5b )^2-\operatorname{Tr} ((\xi^5-\xi^2)b )^2
-\operatorname{Tr} (\xi^{16}b )^2-\operatorname{Tr} ((\xi^{16}-\xi^4)b )^2},$$
which is consistent with our result in Theorem \ref{thm.bent g(x)+x_ix_j}.
\end{example}

To end this section, we point out that more bent functions of the form \eqref{f=g+F} can be obtained from our results. The previous works in this direction focused on the constructions of bent functions satisfying \eqref{f=g+F} and \eqref{g} and new bent functions were obtained in \cite{M,XCX,XCX1,XCX2,WWL,T,QTZF}. Using similar techniques, the analogues of the results in above references can also be obtained for bent functions which satisfy \eqref{f=g+F} and \eqref{g-Aij}, we omit the details here.

\section{Constructions of bent functions from non-bent ones}\label{cons2}

In this section, we aim to construct the bent function $f(x)$ of the form \eqref{f=g+F} from a non-bent function $g(x)$, which is much more difficult than the case when $g(x)$ is a bent function. Let $\tau=2$ and $m,k$ be integers with $n=2m$ and $d=\gcd(k,n)$, we investigate the bentness of the functions having the form
\begin{equation}\label{f.g gold.k=2}
f(x)=\operatorname{Tr}(ax^{p^k+1})+\operatorname{Tr}(ux)\operatorname{Tr}(vx),
\end{equation}
where $a\in\mathbb{F}_{p^n}^{*}$ and $u, v\in\mathbb{F}_{p^n}^{*}$. We always assume that $u\neq v$ if $p=2$ since  $f(x)$ is reduced to $\operatorname{Tr}(ax^{p^k+1})+\operatorname{Tr}(ux)$ when  $p=2$.

\begin{thm}\label{p=2.n/d odd.f}
Let $p=2$, $n/d$ be odd and $u,v$ be two distinct elements in $\mathbb{F}_{2^n}^{*}$. Then the Boolean function $f(x)$ of the form \eqref{f.g gold.k=2} is bent if and only if $d=2$ and $\operatorname{Tr}^{n}_{2}(uc^{-1})\cdot\operatorname{Tr}^{n}_{2}(vc^{-1})
\cdot\operatorname{Tr}^{n}_{2}((u+v)c^{-1})\neq0$,
where $c\in\mathbb{F}_{2^n}^{*}$ is the unique element satisfying $c^{2^k+1}=a$.
\end{thm}

\begin{proof}
Set $p=2$ and $\tau=2$, then for any $b\in\mathbb{F}_{2^n}$, Theorem \ref{thm.wf.f=g+F} yields
\begin{equation}\label{gold1}
 \widehat{f}(b)=\frac{1}{2} (\widehat{g}(b)+\widehat{g}(b+u)+\widehat{g}(b+v)-\widehat{g}(b+u+v) ).
\end{equation}
Notice that $\widehat{g}(b)\in\{0, \pm 2^{\frac{n+d}{2}}\}$ and $\widehat{g}(b)\ne 0$ if and only if $\operatorname{Tr}^{n}_{d}(bc^{-1})=1$ by Lemma \ref{W(g).odd}. Since $n$ is even and $n/d$ is odd, we have $d\geq 2$ is even.

(1): sufficiency. If $d=2$ and $\operatorname{Tr}^{n}_{2}(uc^{-1})\cdot\operatorname{Tr}^{n}_{2}(vc^{-1})
\cdot\operatorname{Tr}^{n}_{2} ((u+v)c^{-1} )\neq0$, then $\operatorname{Tr}^{n}_{2}(uc^{-1})$, $\operatorname{Tr}^{n}_{2}(vc^{-1})$ and $\operatorname{Tr}^{n}_{2} ((u+v)c^{-1} )$ are distinct nonzero elements. Suppose that $\theta$ is a primitive element of $\mathbb{F}_{2^2}$, one then gets
$$ \{\operatorname{Tr}^{n}_{2}(uc^{-1}), \operatorname{Tr}^{n}_{2}(vc^{-1}), \operatorname{Tr}^{n}_{2}((u+v)c^{-1}) \}=\{1,\theta,\theta^2\}.$$
Therefore, for $b\in\mathbb{F}_{2^n}$, one can conclude that
$$ \{\operatorname{Tr}^{n}_{2}(bc^{-1}),\operatorname{Tr}^{n}_{2}((b+u)c^{-1}), \operatorname{Tr}^{n}_{2}((b+v)c^{-1}), \operatorname{Tr}^{n}_{2}((b+u+v)c^{-1}) \}=\{0,1,\theta,\theta^2\}.$$
This together with Lemma \ref{W(g).odd} implies that
$$\widehat{f}(b)=\frac{1}{2}(\pm 2^{\frac{n+2}{2}}+3\cdot 0)=\pm2^m.$$

(2): necessity. Since $\operatorname{Tr}^{n}_{d}(bc^{-1})$ runs through $\mathbb{F}_{2^d}$ when $b$ ranges over $\mathbb{F}_{2^n}$, if $d>2$, one has
$$\{\operatorname{Tr}^n_{d}(uc^{-1})s+\operatorname{Tr}^n_{d}(vc^{-1})t: s,t\in\mathbb{F}_{2}\}\subsetneqq
\{\operatorname{Tr}^{n}_{d}(bc^{-1})+1: b\in\mathbb{F}_{2^n}\}.$$
This indicates that there exists $b\in\mathbb{F}_{2^n}$ such that $\operatorname{Tr}^n_{d} ((b+su+tv)c^{-1} )\neq 1$ for arbitrary $s,\,t\in\mathbb{F}_{2}$ which leads to $\widehat{f}(b)=0$ by \eqref{gold1} and then $f(x)$ is not bent. Now assume that $d=2$ and at least one of $\operatorname{Tr}^{n}_{2}(uc^{-1})$, $\operatorname{Tr}^{n}_{2}(vc^{-1})$ and $\operatorname{Tr}^{n}_{2}((u+v)c^{-1})$ is $0$, without loss of generality, let $\operatorname{Tr}^{n}_{2}(uc^{-1})=0$, then
$$\operatorname{Tr}^n_{2}((b+su+tv)c^{-1})=\operatorname{Tr}^n_{2}(bc^{-1})
+t\operatorname{Tr}^n_{2}(vc^{-1}).$$
Observe that there also exists $b\in\mathbb{F}_{2^n}$ such that $\operatorname{Tr}^n_{2}(bc^{-1})+t\operatorname{Tr}^n_{2}(vc^{-1})\neq 1$ for any $t\in\mathbb{F}_{2}$ since
$$\{\operatorname{Tr}^n_{d}(vc^{-1})t: t\in\mathbb{F}_{2}\}\subsetneqq
\{\operatorname{Tr}^{n}_{d}(bc^{-1})+1: b\in\mathbb{F}_{2^n}\},$$
i.e., there exists $b\in\mathbb{F}_{2^n}$ such that $\widehat{g} (b+su+tv )=0$ for any $s,\,t\in\mathbb{F}_{2}$. Then by \eqref{gold1} one has that $f(x)$ is not bent.
\end{proof}

\begin{example}
Let $p=2$, $n=6$, and $u,\,v$ be two distinct elements in $\mathbb{F}_{2^6}^{*}$. Suppose that $k$ is an integer and $d=\gcd(n,k)$. Taking $a=\xi^{2^k+1}$ and $c=\xi$, where $\xi$ is a primitive element of $\mathbb{F}_{2^6}$, then $c^{2^k+1}=a$.  Magma experiments show that
$f(x)=\operatorname{Tr}(\xi^{2^k+1}x^{2^k+1})+\operatorname{Tr}(ux)\operatorname{Tr}(vx)$ is bent
if and only if $k=2$ or $4$ and $\operatorname{Tr}^{6}_{2}(\xi^{-1}u)\cdot\operatorname{Tr}^{6}_{2}(\xi^{-1}v)
\cdot\operatorname{Tr}^{6}_{2}(\xi^{-1}(u+v))\neq0$.
\end{example}

\begin{thm}\label{thm-non-p}
Let $p$ be an arbitrary prime and $a, u, v\in\mathbb{F}_{p^n}^{*}$.
If $n/d$ is even and $a^{\frac{p^n-1}{p^d+1}}=(-1)^{\frac{m}{d}}$, then the function $f(x)$ of the form \eqref{f.g gold.k=2} is bent if and only if $d=1$, $\operatorname{Tr}^n_{2}(v/(ac^{p^k}))\neq0$ and $\frac{\operatorname{Tr}^n_{2}(u/(ac^{p^k}))}{\operatorname{Tr}^n_{2}(v/(ac^{p^k}))}\notin\mathbb{F}_{p}$, where $c$ satisfies $a^{p^k}c^{p^{2k}}+ac=0$.
\end{thm}

\begin{proof}
If $n/d$ is even and $a^{\frac{p^n-1}{p^d+1}}=(-1)^{\frac{m}{d}}$, then by Lemma \ref{W(g).even} for $b\in \mathbb{F}_{p^n}$ one has that $$\widehat{g}(b)=(-1)^{m/d+1}p^{m+d}\overline{\chi}(ax_0^{p^k+1})$$
if $\operatorname{Tr}^n_{2d}(b/(ac^{p^k}))=0$ and $\widehat{g}(b)=0$ otherwise, where $x_0$ is the solution of $a^{p^k}x^{p^{2k}}+ax=-b^{p^k}$. Set $\tau=2$, then Theorem \ref{thm.wf.f=g+F} yields
\begin{equation}\label{gold2}
\widehat{f}(b)=\frac{1}{p}\sum_{s=0}^{p-1}\sum_{t=0}^{p-1}
\omega^{-st}\widehat{g}(b+su+tv), \forall\; b\in\mathbb{F}_{p^n}.
\end{equation}

(1): sufficiency. If $d=1$, $\operatorname{Tr}^n_{2}(v/(ac^{p^k}))\neq0$ and $\frac{\operatorname{Tr}^n_{2}(u/(ac^{p^k}))}{\operatorname{Tr}^n_{2}(v/(ac^{p^k}))}\notin\mathbb{F}_{p}$, then
$$\operatorname{Tr}^n_{2} ((su+tv)/(ac^{p^k}) )=s\operatorname{Tr}^n_{2}(u/(ac^{p^k}))
+t\operatorname{Tr}^n_{2}(v/(ac^{p^k}))\neq0$$
when $s,\,t\in\mathbb{F}_{p}$ are not 0 simultaneously. This implies that $\operatorname{Tr}^n_{2}((su+tv)/(ac^{p^k}))=0$ only if $(s,\,t)=(0,\,0)$ and $\operatorname{Tr}^n_{2}((su+tv)/(ac^{p^k}))$ are distinct for each pair $(s,t)$ when $s,\,t$ range from 0 to $p-1$. Suppose that $\theta$ is a primitive element of $\mathbb{F}_{p^2}$, one gets
$$\{\operatorname{Tr}^n_{2} ( (su+tv )/(ac^{p^k}) ): 0\leq s,\,t\leq p-1\}= \{0,\,1,\,\ldots,\theta^{p^2-2} \}.$$
Then, for any $b\in\mathbb{F}_{p^n}$, there exists exactly one pair $(s,t)$ such that $\operatorname{Tr}^n_{2}((b+su+tv)/(ac^{p^k}))=0$ when $s,\,t$ run over $\mathbb{F}_{p}$, set as $(s',\,t')$.
Hence, according to Lemma \ref{W(g).even} and \eqref{gold2}, one can claim
$$\widehat{f}(b)=\frac{1}{p}((-1)^{m+1}p^{m+1}\omega^{-s't'}\overline{\chi}(ax_0^{p^k+1})+(p-1)\cdot 0)=(-1)^{m+1}p^{m}\omega^{-s't'}\overline{\chi}(ax_0^{p^k+1}),$$
where $x_0$ is the solution of $a^{p^k}x^{p^{2k}}+ax=- (b+s'u+t'v )^{p^k}$,
i.e., $f(x)$ is bent.

(2): necessity. If $d>1$, then by the property of trace functions, we have
$$\{\operatorname{Tr}^n_{2d}(u/(ac^{p^k}))s+\operatorname{Tr}^n_{2d}(v/(ac^{p^k}))t: s,t\in\mathbb{F}_{p}\}\subsetneqq
\{-\operatorname{Tr}^{n}_{2d}(b/(ac^{p^k})): b\in\mathbb{F}_{p^n}\},$$
which indicates that there exists $b\in\mathbb{F}_{p^n}$ such that $\operatorname{Tr}^n_{2d} ((b+su+tv)/(ac^{p^k}) )\neq 0$ for any $s,\,t\in\mathbb{F}_{p}$. For this $b$, by Lemma \ref{W(g).even} and \eqref{gold2}, one gets $\widehat{f}(b)=0$. Thus, $f(x)$ cannot be bent if $d>1$. Now we assume that $d=1$ and $\operatorname{Tr}^n_{2}(u/(ac^{p^k}))=r\operatorname{Tr}^n_{2}(v/(ac^{p^k}))$ for some $r\in\mathbb{F}_{p}$. Then
$$\operatorname{Tr}^n_{2}((b+su+tv)/(ac^{p^k}))=\operatorname{Tr}^n_{2}(b/(ac^{p^k}))
+(sr+t)\operatorname{Tr}^n_{2}(v/(ac^{p^k})).$$
Since
$$\{(sr+t)\operatorname{Tr}^n_{2}(v/(ac^{p^k})): s,t\in\mathbb{F}_{p}\}\subsetneqq
\{-\operatorname{Tr}^n_{2}(b/(ac^{p^k})): b\in\mathbb{F}_{p^n}\},$$
then there exists $b\in\mathbb{F}_{p^n}$ such that $\operatorname{Tr}^n_{2d} ((b+su+tv)/(ac^{p^k}) )\neq 0$ for arbitrary $s,\,t\in\mathbb{F}_{p}$, and consequently, for this $b$, we have $\widehat{f}(b)=0$ by \eqref{gold2} which implies that $f(x)$ is not bent.
\end{proof}

\begin{example}
Let $p=3$, $n=2m=4$ and $u,\,v\in\mathbb{F}_{3^4}^{*}$. Let $k$ be an integer with $1\leq k\leq 4$ and $d=\gcd(n,k)$. Taking $a=1$ and $c=\xi$ with $\xi^{3^{2k}-1}=-1$, where $\mathbb{F}_{3^4}^*=\langle\xi\rangle$, then $a^{3^k}c^{3^{2k}}+ac=0$.  Magma experiments show that
$f(x)=\operatorname{Tr}(x^{3^k+1})+\operatorname{Tr}(ux)\operatorname{Tr}(vx)$ is bent
if and only if $k=1$ or 3 and $\operatorname{Tr}^{4}_{2}(\xi^{-p^k}u)/\operatorname{Tr}^{4}_{2}(\xi^{-p^k}v)\notin\mathbb{F}_{3}$.
\end{example}

\begin{remark}
The construction of the bent function $f(x)$ with the form \eqref{f=g+F} from a non-bent function $g(x)$ is much more difficult when $\tau\geq 3$. Our computer experiments indicate that such bent functions indeed exist and then it will be interesting to find an efficient way to construct this kind of bent functions.
\end{remark}

\begin{table}[ht]
\setlength{\belowcaptionskip}{-0.2cm}
\caption{Known bent functions over $\mathbb{F}_{p^n}$ with the form \eqref{f=g+F}} \label{Tab-f}
\footnotesize
\begin{center}
\setlength{\tabcolsep}{0.3mm}{
\begin{tabular}{|c|c|c|c|}
\hline \text { $p$ } & \text { $g(x)$ is bent with $\widetilde{g}(x)$ satisfies \eqref{g-Aij}} & \text { $F(x_1,\cdots,x_{\tau})$,\,$\tau\geq2$ } & \text { Refs. }\\
\hline
\hline $p=2$ & $A_{12}=0$ & $F(x_1,x_2)=x_1x_2$& \cite{M}\\
\hline $p=3$ & $A_{11}=0,\,A_{12}=0$ & $F(x_1,x_2)=x_1x_2$& \cite{XCX1}\\
\hline $p=3$ & $(A_{12}-1)^2\neq A_{11}A_{22}$& $F(x_1,x_2)=x_1x_2$& \cite{XCX2}\\
\hline $p=2$ & $A_{12}=A_{13}=A_{23}=0$& $F(x_1,x_2,x_3)=x_1x_2x_3$ &\cite{XCX}\\
\hline $p=2$ & $A_{ij}=0,\,1\leq i< j\leq \tau$ & $F(x_1,\cdots,x_{\tau})=x_1\cdots x_{\tau}$& \cite{WWL}\\
\hline $p=2$ & $A_{ij}=0,\,1\leq i< j\leq \tau$ &{\rm any}& \cite{T}\\
\hline odd $p$ & $A_{ij}=0,\,1\leq i\leq j\leq \tau$&  {\rm any}& \cite{QTZF}\\
\hline $p=2$ & $A_{ij}\neq0,\,i,j\in\Gamma,\,\#\Gamma=2\ell$, otherwise $A_{ij}=0$&$\deg(F(x_1,\cdots,x_{\tau})|_{x_i=h_i, i\not\in \Gamma})\leq1$& Thm.\,\ref{p=2.F1}\\
\hline $p=2$ & $A_{i_1j_s}\neq0,\,s=1,\cdots,\ell_1$, otherwise $A_{ij}=0$&
$\deg(F(x_1,\cdots,x_{\tau})|_{x_i=h_i, i\not\in \Gamma, x_{j_s}=x_{j_1}+h_{j_1}+h_{j_s}, s\neq1})\leq 1$
&Thm.\,\ref{p=2.F2}\\
\hline odd $p$ & $(A_{12}-1)^2-4A_{11}A_{22}\neq 0$& $F(x_1,x_2)=x_1x_2$ &Prop.\,\ref{cor.bent g(x).k=2}\\
\hline odd $p$ & $A_{\tau_1\tau_1},A_{\tau_2\tau_2},A_{\tau_1\tau_2}\in\mathbb{F}_{p}$, otherwise $A_{ij}=0$ &$ \deg(F(x_1,\cdots,x_{\tau})|_{x_i=g_i(b), i\not\in \Gamma})\leq2$ &Thm.\,\ref{thm.bent g(x)}\\
\hline odd $p$ & $A_{ii}\neq0,\,i=1,\cdots,\tau$, otherwise $A_{ij}=0$ &$F(x_1,\cdots,x_{\tau})=\sum_{1\leq i\leq j \leq \tau}a_{ij}x_ix_j$&Thm.\,\ref{thm.bent g(x)+x_ix_j}\\
\hline
\hline \text { $p$ } & \text { $g(x)$ is non-bent} & \text { $F(x_1,\cdots,x_{\tau})$,\,$\tau\geq2$ } & \text { Refs. }\\
\hline
\hline $p=2$ & $\operatorname{Tr}(ax^{2^k+1})$, $n/\gcd(k,n)$ is odd& $F(x_1,x_2)=x_1x_2$& Thm.\,\ref{p=2.n/d odd.f}\\
\hline any $p$ & $\operatorname{Tr}(ax^{p^k+1})$, $n/\gcd(k,n)$ is even& $F(x_1,x_2)=x_1x_2$ &Thm.\,\ref{thm-non-p}\\
\hline
\end{tabular}}
\end{center}
- where $A_{ij}$, $g_i(b)$ are defined by \eqref{g-Aij} with $h_i=g_i(b)+A_{ii}$,  $\Gamma$ and $F(x_1,\cdots,x_{\tau})$ are given by \eqref{T} and \eqref{F(x_1,...,x_n)}.
\end{table}

\section{Conclusion}\label{conc}

In this paper, we investigated the bentness of the function $f(x)$ over the finite field $\mathbb{F}_{p^n}$ with the form \eqref{f=g+F} by using different kinds of $g(x)$ as before (see Table \ref{Tab-f}), where $n$ is a positive integer and $p$ is a prime. We firstly obtained a generic result on the Walsh transform of the function $f(x)$, which generalized some previous works, and then characterized its bentness for the case $g(x)$ is bent for $p=2$ and $p>2$ respectively. It was shown that bent functions with the maximal algebraic degree can be obtained from our construction. Moreover, we presented a class of bent functions $f(x)$ of the form \eqref{f=g+F} when $g(x)$ is a non-bent Gold function.

\section*{Acknowledgements}

This work was supported in part by the National Natural Science Foundation of China (Nos. 62072162, 61761166010, 12001176, 61702166), the Application Foundation Frontier Project of Wuhan Science and Technology Bureau (No. 2020010601012189) and the National Key Research and Development Project (No. 2018YFA0704702).

\end{document}